\documentclass[a4paper,12pt]{amsart}
\usepackage[lmargin=2.4cm,rmargin=2.25cm,bmargin=3cm,tmargin=2.5cm]{geometry}

\usepackage[utf8]{inputenc}
\usepackage{amsmath} 
\usepackage{amssymb}
\usepackage{amsthm}
\usepackage{textcomp}
\usepackage{enumerate}
\usepackage[sort,numbers]{natbib}
\usepackage{array}
\usepackage{verbatim}
\usepackage{hyperref}
\usepackage{graphicx}
\usepackage{booktabs}
\usepackage{algorithm}
\usepackage{algpseudocode}
\usepackage{stmaryrd} 
\usepackage{xcolor}

\usepackage{tcolorbox}
\usepackage{soul}

\newcommand{\defeq}{\mathrel{\mathop:}=}

\newcommand{\F}{\mathcal{F}}

\newcommand{\X}{\mathsf{X}}

\newcommand{\M}{\mathbb{M}}

\newcommand{\uarg}{\,\cdot\,}
\newcommand{\ud}{\mathrm{d}}
\newcommand{\R}{\mathbb{R}}
\newcommand{\real}{\mathbb{R}}
\newcommand{\N}{\mathbb{N}}
\renewcommand{\P}{\mathbb{P}}
\newcommand{\E}{\mathbb{E}}

\newcommand{\charfun}[1]{\mathbf{1}\left(#1\right)}

\newcommand{\given}{\,:\,}

\renewcommand{\vec}[1]{\boldsymbol{#1}}

\newcommand{\ind}[1]{{#1}}

\theoremstyle{plain}
\newtheorem{theorem}{Theorem}

\newtheorem{proposition}[theorem]{Proposition}
\newtheorem{lemma}[theorem]{Lemma}

\newtheorem{definition}[theorem]{Definition}
\newtheorem{assumption}[theorem]{Assumption}

\theoremstyle{remark}
\newtheorem{remark}[theorem]{Remark}

\renewcommand{\L}{{\mathcal L}}
\newcommand{\Lmut}{{\mathcal L}^{\textrm{mut}}}
\newcommand{\Ljump}{{\mathcal L}^{\textrm{jump}}}

\newcommand{\XX}{{\mathbb X}}
\newcommand{\FF}{{\mathbb F}}
\newcommand{\QQ}{{\mathcal Q}}
\newcommand{\testf}{C^\infty_c}

\newcommand{\chrf}{\mathbf{1}}

\begin{document}

\title[Particle filters with weakly informative observations]{On resampling schemes for particle filters with weakly informative observations}

\author{Nicolas Chopin}
\address[NC]{ENSAE, Institut Polytechnique de Paris}

\author{Sumeetpal S.~Singh}
\address[SSS]{Department of Engineering, University of Cambridge}

\author{Tomás Soto}
\address[TS]{School of Engineering Science, LUT University}

\author{Matti Vihola}
\address[MV]{Department of Mathematics and Statistics, University of Jyväskylä}

\begin{abstract} 
  We consider particle filters with weakly informative observations (or
  `potentials') relative to the latent state dynamics. The particular focus of this work is 
  on particle filters to approximate time-discretisations of continuous-time Feynman-Kac path integral models --- a scenario that naturally arises when addressing filtering and smoothing problems in continuous time --- but our findings are indicative about weakly informative settings beyond this context too. We study the performance of different resampling schemes, such as systematic resampling, SSP (Srinivasan
  sampling process) and stratified resampling, as the time-discretisation becomes finer and also identify their continuous-time limit, which is expressed as a suitably defined `infinitesimal generator.' By contrasting these generators, we find that (certain modifications of) systematic and SSP resampling `dominate' stratified and independent `killing' resampling in terms of their limiting overall resampling rate. The reduced intensity of resampling manifests itself in lower variance in our numerical experiment. This efficiency result, through an ordering of the resampling rate, is new to the literature. The second major contribution of this work concerns the analysis of the limiting behaviour of the entire population of particles of the particle filter as the time discretisation becomes finer. We provide the first proof, under general conditions, that the particle approximation of the discretised continuous-time Feynman-Kac path integral models converges to a (uniformly weighted) continuous-time particle system. 
  \end{abstract} 


\maketitle
\sloppy

\section{Introduction} 

Particle filters \cite{gordon-salmond-smith} have become a workhorse of
non-linear stochastic filtering and statistical state space modelling. The heart
of the particle filters is the `interaction' within the particles, which is
caused by the resampling (or selection) step of the algorithm. See e.g.
\cite{chopin-papaspiliopoulos} for a general introduction to particle filtering
and its applications in various scientific fields. 

When the observations (or more generally the potential functions) are
informative, that is, the weights that go into the resampling have a high
variability, empirical evidence suggests that the choice of the resampling
strategy can only make a small difference. In contrast, when the observations
(potential functions) are weakly informative and the weights tend to be close to
uniform, the performance differences between different resampling methods can be
substantial. In the weakly informative regime, it is therefore important to use
an appropriate resampling scheme.

Different resampling schemes have been suggested in the literature, and some
resampling schemes have been compared in terms of the conditional variance that
they introduce to the weights \cite{douc-cappe}.  More recently,
\cite{gerber-chopin-whiteley} considered ordering of resampling schemes with
respect to a so-called negative association, and introduced a new `SSP'
resampling scheme, which is preferable based on their theoretical findings.
These analyses, as most other theoretical analyses on resampling methods, focus
on the asymptotic regime in the number of particles $N\to\infty$. 

We do not consider the asymptotic $N\to\infty$, but keep $N$ fixed instead. In
contrast, we consider the asymptotic behaviour of the resampling schemes as the
potentials become less and less informative. One domain of applications, where
such a situation naturally arises, is time-discretisations of continuous-time
particle systems \citep{delmoral-miclo-moran,delmoral-miclo-seminaire} that
approximate so-called Feynman-Kac path integrals. We study the behaviour of
discrete-time resampling schemes when applied with discretisations of
continuous-time Feynman-Kac path integral models.

The main contributions of this work are:
\begin{itemize}
\item We introduce a condition for (discrete-time) resampling, namely Assumption \ref{a:resampling-generator}, which ensures the existence of a limiting continuous-time particle system. In particular, when this assumption is satisfied and under further general conditions, Theorem \ref{thm:convergence} establishes the limiting continuous-time particle system that a particle implementation of the time-discretised Feynman-Kac path integral model converges to (as the discretisation is refined).
Having established this asymptotic limit, Theorem \ref{thm:fk-unbiased} gives a result on its use for unbiased estimation of certain expectations with respect to the continuous-time Feynman-Kac path integral.  

\item In Section \ref{sec:resampling-schemes} we then proceed to analyse which resampling schemes satisfy Assumption \ref{a:resampling-generator} in a series of results, namely Propositions \ref{pr:killing-resampling-rate}, \ref{prop:stratified-generator},  \ref{prop:systematic-generator} and \ref{prop:ssp-intensity}.
When the condition holds, resampling schemes can be ordered by comparing their limiting continuous-time resampling intensities in Theorem \ref{thm:resampling-order} and Proposition \ref{prop:symmetrised-systematic-intensity}. We find that certain variants of systematic resampling and SSP resampling have a common limiting overall resampling intensity, which is guaranteed to be lower than that of the stratified or so-called killing resampling. This suggests that our variants systematic/SSP resampling can be preferable.

\item Our empirical findings (Section \ref{sec:experiments}) about the practical performance are in line with our
theoretical results, and indicate that SSP resampling and systematic
resampling, when applied after a prior partial ordering of the weights about
their mean, lead to the best performing particle filters. This complements the
positive findings for SSP \cite{gerber-chopin-whiteley}, and suggests that a
partial ordering of weights should always be used with systematic resampling.
\end{itemize}
Overall our results fill important gaps in the literature on particle filters,
in particular concerning their continuous-time limiting behaviour.
This is in contrast with e.g. 
\cite{del-moral-jacob-lee-murray-peters, rousset, arnaudon-delmoral}, 
who considered directly particular (theoretical) continuous-time algorithms (based
on killing resampling), and how they may be discretised. 

We consider only resampling schemes that result in a uniformly weighted sample,
so that the particle system remains unweighted. Furthermore, all the studied
resampling schemes lead to `single-event' continuous-time limits, meaning that
at most one particle can disappear at an individual time. These exclude the
popular alternative strategy, adaptive resampling \cite{liu-chen-blind}, in
which resampling is triggered at certain (random) times.

\section{Hidden Markov models and particle filters}
\label{sec:notation} 

A hidden Markov model (HMM) consists of two components: a latent (unobserved) Markov chain $X_{1:T} = (X_{1},\ldots,X_T)$ on state a space $\X$ with an initial probability density $f_1(x_1)$ and transition densities $f_k(x_k\mid x_{k-1})$; and conditionally independent observations $Y_{1:T}$ with conditional laws $g_k(y_k\mid x_k)$. The particle filter can be used to estimate integrals with respect to a conditional probability law, the so-called smoothing distribution:
\begin{equation*}
  p(x_{1:T}\mid y_{1:T}) \propto p(x_{1:T},y_{1:T}) = f_1(x_1) g_1(y_1\mid x_1) \prod_{k=2}^T f_k(x_k\mid x_{k-1}) g_k(y_k \mid x_k).
\end{equation*}
The Feynman-Kac model is an abstraction and generalisation which allows for defining a family of unnormalised probability densities $\gamma(x_{1:T})$ which is equivalent to $p(x_{1:T},y_{1:T})$ in the HMM context. It is based on 'proposal' Markov chain laws $M_1(x_1)$ and $M_k(x_k\mid  x_{k-1})$ \cite{del-moral} and non-negative 'potential functions' $G_k(x_{1:k})$, where $x_{1:k}\in \X^k \rightarrow G_k(x_{1:k})$ (which can implicitly depend on $y_{1:T}$ too), so that 
$p(x_{1:T}\mid y_{1:T}) = \pi_T(x_{1:T})$, with $\pi_T$ defined as follows: 
\begin{equation}
\pi_k(x_{1:k}) \defeq \frac{\gamma_k(x_{1:k})}{\mathcal{Z}_k} \quad
\text{where}\quad \gamma_k(x_{1:k})\defeq M_1(x_1) G_1(x_1)
\prod_{j=2}^k M_j(x_j\mid x_{j-1}) G_j(x_{1:j})
\label{eq:target}
\end{equation}
and
$\mathcal{Z}_k\defeq \int \gamma_k(x_{1:k})\ud x_{1:k}$.
In the HMM context, we typically set $G_1(x_1) = g_1(y_1\mid x_1) f_1(x_1)/M_1(x_1)$ and
$$
G_k(x_{1:k}) 
= g_k(y_k\mid x_k) f_k(x_k\mid x_{k-1})/M_k(x_k\mid x_{k-1}).
$$
In the simplest case, when $M_k \equiv f_k$, we get $G_k(x_{1:k})  = g_k(y_k\mid x_k)$. Henceforth, the domain of the potential $G_k$ will be apparent from its argument. So when $G_k:\X \rightarrow [0,\infty)$, an instance of which is the just mentioned simplest case, we will write $G_k(x_k)$.

The focus of this paper is in situations where $G_k$ are 'weakly informative,' that is, when $G_k(x_k)$ is nearly constant for typical values $x_k$ (with respect to $\pi_k$). In the HMM setting, this typically occurs when the observations $Y_k \sim g_k(\uarg \mid x_k)$ have substantial variability compared to the variability of $X_k \sim f_k(\uarg\mid x_{k-1})$, and can also occur when $M_k$ correspond to an approximation of the smoothing distribution \citep[cf.][]{vihola-helske-franks}.
However, our main theoretical framework is beyond the HMM context, where $(M_{1:T},G_{1:T})$ correspond to time-discretisations of a continuous-time path integral model (Section \ref{sec:path-integral}). 

Hereafter, we will focus on the Feynman-Kac model in \eqref{eq:target}, assuming that $\pi_T$ is well-defined, that is, the normalising constant $\mathcal{Z}_T$ is finite and strictly positive. We use the notation $a{:}b = (a,a+1,\ldots,b)$ for integers $a\le b$, and use the same notation for indexing and double indexing of sequences. Thus for sequences $\{x_i\}_i$, $\{y^j\}_j$ and $\{z_i^j\}_{i,j}$ we write $x_{a:b} \defeq (x_a,\ldots,x_b)$, $y^{\ind{a:b}} \defeq (y^{\ind{a}},\ldots,y^{\ind{b}})$ and $z_{a:b}^{\ind{j_{a:b}}} \defeq (z_a^{\ind{j_a}},\ldots,z_b^{\ind{j_b}})$. For $N \in \N$, we denote $[N] = \{1,2,\cdots,N\}$. The sequence $1{:}N$ with $k$ omitted and $\ell$ duplicated is denoted as $[k\to \ell]_N \defeq \big(1{:}(k-1), (k+1){:}\ell, \ell{:}N\big)$. The notation `$\ud x$' implicitly stands for a $\sigma$-finite dominating measure on $\mathsf{X}$, integers are equipped with the counting measure, product spaces are equipped with products of dominating measures and test functions are implicitly assumed to be measurable.

Let us then turn to the particle filter algorithm based on the Feynman-Kac
model. The particle filter involves one additional ingredient: the resampling
mechanism, which is determined by a probability distribution $r(\uarg\mid
g^{\ind{1:N}})$ on $[N]^N$, given unnormalised weights $g^{\ind{1:N}}\in
[0,\infty)^N$. We only consider resampling schemes $r$, which satisfy the
following condition (which may be traced back to (4) in
\cite{crisan-delmoral-lyons}):
\begin{assumption}
  \label{a:unbiased-resampling} 
  Whenever $\sum_{i=1}^N g^{\ind{i}} > 0$, 
  the indices $A^{\ind{1:N}} \sim r(\uarg\mid g^{\ind{1:N}})$ satisfy
\begin{equation}
    \E\bigg[\frac{1}{N}\sum_{i=1}^N \charfun{A^{\ind{i}} = j}\bigg] =
\frac{g^{\ind{j}}}{ \sum_{i=1}^N g^{\ind{i}} }
\label{eq:res-unbiased}
\end{equation}
  for all $j\in[N]$. 
\end{assumption}
A resampling method that satisfies this assumption is known as being {\it unbiased} \cite{andrieu-doucet-holenstein} since the expected number occurrences of outcome $j$ in the population $A^{\ind{1:N}}$ is $Ng^{\ind{j}}/( \sum_{i=1}^N g^{\ind{i}} )$.
Algorithm \ref{alg:pf} presents the particle filter in pseudo-code.
\begin{algorithm}
  \caption{\textsc{ParticleFilter}$(M_{1:T}, G_{1:T-1}, r, N)$}
  \label{alg:pf} 
\begin{algorithmic}[1]
\State
Draw $X_1^{\ind{i}} \sim M_1(\uarg)$ and set $\vec{X}_1^{\ind{i}} = X_1^{\ind{i}}$ for $i\in[N]$.
\For{$k=1,\ldots,T-1$}
\State Draw $A_{k}^{\ind{1:N}} \sim r\big(\uarg \mid
G_{k}(\vec{X}_{k}^{\ind{1}}),\ \dots, G_{k}(\vec{X}_{k}^{\ind{N}})\big)$
\State Draw $X_{k+1}^{\ind{i}} \sim M_{k+1}(\uarg \mid
          X_{k}^{\ind{A_{k}^{\ind{i}}}})$ and set 
          $\vec{X}_{k+1}^{\ind{i}} 
          = (\vec{X}_{k}^{\ind{A_k^{\ind{i}}}}, X_{k+1}^{\ind{i}})$ for $i\in[N]$.
\EndFor
\State \textbf{output}
$\big((X_1^{\ind{1:N}},\ldots,X_T^{\ind{1:N}}),(A_1^{\ind{1:N}},\ldots, A_{T-1}^{\ind{1:N}}))\big)$
\end{algorithmic}
\end{algorithm}

With the shorthand $\underline{x}_k = x_k^{\ind{1:N}}$ for the whole particle system, we may write the law of the output $(\underline{X}_{1:T},\underline{A}_{1:T-1})$ of Algorithm \ref{alg:pf} in the following form:
\begin{equation}
\zeta^{(N)}(\underline{x}_{1:T}, \underline{a}_{1:T-1})
  =
  \bigg(\prod_{i=1}^N M_1(x_1^{\ind{i}}) \bigg)
  \bigg(\prod_{k=1}^{T-1} 
   r\big(\underline{a}_{k}\mid
  G_{k}(\underline{\vec{x}}_{k} )\big)
  \prod_{i=1}^N M_{k+1}(x_{k+1}^{\ind{i}}\mid x_{k}^{\ind{a_{k}^{\ind{i}}}}) \bigg).
  \label{eq:pf-law} 
\end{equation}
As in Algorithm \ref{alg:pf}, denote $\vec{x}_1^{\ind{i}} = x_1^{\ind{i}}$ and $\vec{x}_{k+1}^{\ind{i}} = (\vec{x}_{k}^{a_k^{\ind{i}}}, x_{k+1}^{\ind{i}})$. We have used the shorthand $G_{k}(\underline{\vec{x}}_{k} )$ in the second argument of $r(\cdot \mid \cdot)$ to mean $G_{k}({\vec{x}}_{k}^{\ind{1}}),\ldots, G_{k}({\vec{x}}_{k}^{\ind{N}})$.

Under Assumption \ref{a:unbiased-resampling}, the output of Algorithm
\ref{alg:pf} satisfies the following unbiasedness condition \cite[][Theorem
7.4.2]{del-moral}, which is key for particle Markov chain Monte Carlo
\citep{andrieu-doucet-holenstein}:
\begin{equation}
    \E_{\zeta^{(N)}}\bigg[\bigg(\prod_{k=1}^T \frac{1}{N} \sum_{i=1}^N G_k(\mathbf{X}_k^{\ind{i}})\bigg)
        \sum_{i=1}^N W_T^i
    f(\mathbf{X}_T^{\ind{i}})\bigg]
    = \gamma_T(f) \defeq \int f(x_{1:T}) \gamma_T(x_{1:T}) \ud x_{1:T},
    \label{eq:unbiased-smoothing}  
\end{equation}
where $W_T^i \defeq G_T(\mathbf{X}_T^{\ind{i}})/\sum_{j=1}^{N}
G_T(\mathbf{X}_T^{\ind{j}})$. In addition, under further conditions on
$M_{1:T}$, $G_{1:T}$ and $f$, the following consistency result  holds
(see e.g. Chapter 11 of \cite{chopin-papaspiliopoulos} 
and references therein, e.g. \cite{crisan-delmoral-lyons}): 
\begin{equation}
    \label{eq:normalised-smoothing}
    \sum_{i=1}^N W_T^i
    f(\mathbf{X}_T^{\ind{i}}) \xrightarrow{N\to\infty} \pi_T(f) \defeq \int f(x_{1:T}) \pi_T(x_{1:T}) \ud x_{1:T}\qquad\text{in probability.}
\end{equation}

\section{Continuous-time path integral model} 
\label{sec:path-integral} 

Continuous-time Feynman-Kac path integral models are the continuous-time
analogue of hidden Markov models discussed above. The smoothing distribution is
defined in terms of expectations of real-valued test functions $\phi$ on the
path space  (more precisely, the Skorohod space $D_\X[0,\tau]$ of c\`adl\`ag
paths of a separable metric space $\X$):
\begin{multline}
    \Pi(\phi) \defeq \frac{1}{\mathcal{Z}_\M} \E_{\M}\bigg[
        \phi\big(Z_{[0,\tau]}\big)
    \exp\Big(-\int_0^\tau V_u(Z_u) \ud u\Big) \bigg] \\
    \quad \text{with}\quad
    \mathcal{Z}_\M \defeq \E_\M\bigg[ \exp\Big(-\int_0^\tau V_u(Z_u) \ud u\Big) \bigg],
    \label{eq:path-integral}
\end{multline}
where $(V_u)_{0 \leq u \leq \tau}$ is a sufficiently regular family of
non-negative potential functions on $\X$, and $\M$ is the law of a Markov
process $Z_{[0,\tau]} \defeq {(Z_u)}_{0 \leq u \leq \tau}$ on $\X$.

We focus on an approximation of the law
\eqref{eq:path-integral} based on a time-discretisation
$0=t_1<\cdots<t_T = \tau$
of the form
\begin{equation}
    \Pi_T(\phi) \defeq
    \frac{1}{\mathcal{Z}_T} \E_{\M}\bigg[ \phi\big(\hat{Z}_{[0,\tau]}\big)
    \prod_{k=1}^{T-1} G_k(Z_{t_{k}},Z_{t_{k+1}})
     \bigg],
    \qquad \mathcal{Z}_T \defeq \E_{\M}\bigg[
    \prod_{k=1}^{T-1} G_k(Z_{t_{k}},Z_{t_{k+1}})
    \bigg],
    \label{eq:discretised-path-integral}
\end{equation}
where $\hat{Z}_u \defeq \sum_{k=1}^{T-1} \charfun{u\in[t_{k},t_{k+1})} Z_{t_k} + \charfun{u = t_T} Z_{t_T}$
is a c\`adl\`ag extension of the skeleton $(Z_{t_1},\ldots,Z_{t_T})$ and
where the potential functions $G_k(Z_{t_{k}},Z_{t_{k+1}})\ge 0$ are approximations of
$\exp\big(-\int_{t_{k}}^{t_{k+1}}
  V_u(Z_u) \ud u\big)$ that can depend only on the values of $Z_u$ at
times $t_{k}$ and $t_{k+1}$. Our theoretical focus is on a simple
Euler-type form
$G_k(Z_{t_{k}},Z_{t_{k+1}}) \defeq G_k(Z_{t_k}) \defeq \exp\big(
-({t_{k+1}}-t_{k})V_{t_{k}}(Z_{t_{k}})\big)$, but
our method is also applicable to other approximation schemes. 
For the fixed time-discretisation $0=t_1<\cdots<t_T = \tau$, we may
define $M_1$ as the initial distribution of $Z_{t_1}$; $M_k$ for $2 \leq k \leq T$ as (an approximation of) the
conditional distribution of $Z_{t_k}\mid Z_{t_{k-1}}$;
$G_1\equiv 1$; and for $1 \leq k\leq T-1$, $G_k(Z_{t_k})$ as just defined. 

Having just declared the Markov kernels $M_{1:T}$ and the potential functions $G_{1:T-1}$, we can now employ Algorithm \ref{alg:pf} to form a particle approximation of $\Pi_T(\cdot)$  in \eqref{eq:discretised-path-integral} using \eqref{eq:normalised-smoothing}; and also an unbiased approximation of its normalising constant $\mathcal{Z}_T$ using \eqref{eq:unbiased-smoothing}. Write
\[
  X^{(T)} \defeq \bigl((X^{(T)}_k)^{\ind{1:N}}\bigr)_{1 \leq k \leq T}
\]
for the resulting particle system, where the superscript $^{(T)}$ refers to the discretisation $(t_k)_{1 \leq k \leq T} \subset [0,\tau]$. Our main focus in Sections \ref{sec:convergence} and \ref{sec:unbiasedestimation} is to study the particle approximations as the discretisation is refined, that is as $T \to \infty$, but for a fixed $N$. We now give a flavour of the results.

The first question is whether the law of the discrete-time particle system converges in some sense to a law of the form \eqref{eq:path-integral} corresponding to a continuous-time Feynman-Kac model. To enable this convergence study, in Section \ref{sec:resampling-schemes} below, we introduce a stability condition (Assumption \ref{a:resampling-generator}) that encompasses several known unbiased resampling schemes. In Sections \ref{sec:convergence} and \ref{sec:unbiasedestimation},  we then present the main convergence results (Theorems \ref{thm:convergence} and \ref{thm:fk-unbiased}) for the particle approximations in the context of It\^o diffusions for $\M$ (see \eqref{eq:sde}). First, Theorem \ref{thm:convergence} studies the convergence of the continuous-time extension of the population of particles $((\hat{X}^{(T)}_{t})^{1:N})_{0\leq t \leq \tau}$, with c\`adl\`ag paths in $(\R^d)^N$ defined by
\[
(\hat{X}^{(T)}_{u})^i \defeq \sum_{k=1}^{T-1} \charfun{u\in[t_{k},t_{k+1})} (X^{(T)}_{k})^i + \charfun{u = t_T} (X^{(T)}_{T})^i, \qquad i \in [N].
\]
In particular, parts (i) and (ii) of Theorem \ref{thm:convergence} identify a continuous-time Markov process with c\`adl\`ag paths in $(\R^d)^N$ that $((\hat{X}^{(T)}_{t})^{1:N})_{0\leq t \leq \tau}$ converges to with respect to finite-dimensional distributions as $T\to\infty$.

Then in Theorem \ref{thm:fk-unbiased} (result \eqref{eq:fk-unbiased}) we show that a certain unnormalised time-marginal
of this limiting continuous-time Markov process  coincides with the unnormalised time-marginal of $\mathcal{Z}_\M \times \Pi(\cdot)$ in \eqref{eq:path-integral}.
Using part (iii) of Theorem  \ref{thm:convergence} and result \eqref{eq:fk-unbiased} of Theorem \ref{thm:fk-unbiased}, we then conclude that the particle approximation of the (unnormalised) time-marginal Feynman-Kac path integral converges (for a fixed $N$) as the discretisation is refined:
\[
\lim_{T \rightarrow \infty}\E \left\{ \left(\frac{1}{N} \sum_{i=1}^N f\bigl( (\hat{X}^{(T)}_{\tau})^i\bigr) \right) 
    \exp\left[-\int_0^\tau  \left( \frac{1}{N} \sum_{i=1}^N V_u\bigl((\hat{X}^{(T)}_{u})^i\bigr) \right) \ud u\right] \right\} =\Pi(f)\times \mathcal{Z}_\M,
\]
where $f:\R^d\rightarrow\R$ is bounded and continuous, and $\Pi(f)$ and $\mathcal{Z}_\M$ are as in \eqref{eq:path-integral} with $f(Z_{[0,\tau]}) \defeq f(Z_{\tau})$. Note that the integrals in the exponential term of the left hand side are easy to evaluate as the $(\hat{X}^{(T)})^i$ are piecewise constant c\`adl\`ag paths.

\section{Discretisation stable resampling schemes} 
\label{sec:resampling-schemes}

Our main focus is on resampling schemes which lead to a \emph{valid
  continuous-time limit} under infinitesimally refined discretisations. It turns
out that the following condition naturally ensures a such continuous-time limit
of a particle filter, and can provide insight about different resampling
schemes.
\begin{assumption}
\label{a:resampling-generator}
For all $v^{\ind{1:N}} \in [0,\infty)^N$ and for all
$a_{1:N} \in [N]^N\setminus\{1{:}N\}$, the
limit
\begin{equation}\label{eq:resampling-generator}
    \lim_{\Delta \to 0^+} \frac{1}{\Delta} r\Big(a^{\ind{1:N}} \;\Big|\; \big(\exp(-\Delta
      v^{\ind{1}}),\ldots, \exp(-\Delta v^{\ind{N}})\big)\Big)
    =: \iota(a^{\ind{1:N}},v^{\ind{1{:}N}}).
\end{equation}
exists, and for any $v^* > 0$ the term inside the left-hand side limit is uniformly bounded for $v^{\ind{1{:}N}} \in [0,v^*]^N$ and $\Delta \in (0,1)$.
\end{assumption} 

The limiting quantity $\iota(a^{\ind{1{:}N}}, v^{\ind{1{:}N}})$ can be interpreted as the \emph{resampling intensity} corresponding to the configuration $a^{\ind{1:N}} \neq 1{:}N$, with instantaneous potential values $v^{\ind{1{:}N}}$. It can be thought of as the `infinitesimal generator' stemming from the resampling $r$ in the continuous-time limit.
The sum of all resampling configurations 
\[
    \iota^*(v^{\ind{1:N}}) \defeq \sum_{a^{\ind{1:N}}\neq 1{:}N} \iota(a^{\ind{1:N}},v^{\ind{1:N}})
\]
is the \emph{overall resampling rate}, that is, the
intensity of any `event' $a^{\ind{1:N}} \neq 1{:}N$.

The basic and popular multinomial resampling scheme, which may be traced back to 
\citep{gordon-salmond-smith}, does not admit a continuous-time limit: 
the probability of survival, that is, getting any permutation of
$1{:}N$, does not tend to unity as $\Delta\to 0$. The same holds for the residual
resampling introduced by \citep{higuchi, liu-chen}. 

Perhaps the simplest scheme which satisfies this condition is a discrete-time
version of the `killing' resampling \citep{del-moral-2013}. In discrete-time
killing, the particle at index $i$ `survives' with probability proportional to
the unnormalised weight $g^{\ind{i}}$, and otherwise will be replaced with any
other particle $j$, with probabilities proportional to $g^{\ind{j}}$. We focus
in particular on the following version of discrete-time killing:
\begin{definition}[Killing resampling]
  \label{def:killing}
\begin{align*}
    r_\mathrm{killing}(a^{\ind{1:N}} \mid g^{\ind{1:N}})
    &\defeq \prod_{i=1}^N \bigg[\charfun{a^{\ind{i}} = i}
    \frac{g^{\ind{i}}}{g^*}
    + \Big(1- \frac{g^{\ind{i}}}{g^*}\Big)
    \sum_{j=1}^N \charfun{a^{\ind{i}} = j}
    \frac{g^{\ind{j}}}{\sum_{\ell=1}^N g^{\ind{\ell}}}
    \bigg],
\end{align*}
where $g^*  = \max_{i \in [N]} g^{\ind{i}}$.
\end{definition}
In fact, any $g^*$ such that $g^{\ind{i}}/g^*\in[0,1]$ for all $i\in[N]$ yields
a valid unbiased resampling. The choice of $g^*$ above, which was also used in the algorithmic `rejection' variant of \cite{murray-lee-jacob}, ensures the highest survival probability, that is, $a^{\ind{1:N}}=1{:}N$. The following result can be verified by a direct calculation.

\begin{proposition}\label{pr:killing-resampling-rate}
  Killing resampling satisfies Assumptions \ref{a:unbiased-resampling} and
  \ref{a:resampling-generator} and has limiting intensity
\[
    \iota_\mathrm{killing}(a^{\ind{1:N}},v^{\ind{1:N}}) =
    \begin{cases}
        \frac{1}{N} (v^{\ind{i}} - v_{\min}) & \text{if } 
        a^{\ind{i}} \neq i,\, a^{\ind{\neg i}} = \neg i, \\
        0 & \text{otherwise},
    \end{cases}
\]
where $\neg i \defeq (1,\ldots,i-1,i+1,\ldots,N)$ and
where $v_{\min}  \defeq \min_{j\in[N]} v^{\ind{j}}$.
Consequently,
\[
    \iota^*_\mathrm{killing}(v^{\ind{1:N}})
    = \frac{N-1}{N} \sum_{i=1}^N (v^{\ind{i}} - v_\mathrm{min})
    = (N-1)(\bar{v} - v_{\min}),
\]
where $\bar{v} \defeq N^{-1} \sum_{i=1}^N v^{\ind{i}}$ is the mean of the potential values $v^{\ind{1{:}N}}$.
\end{proposition}

\subsection{Stratified and systematic resampling} 

For the rest of Section \ref{sec:resampling-schemes}, we assume fixed unnormalised weights $g^{\ind{1{:}N}}\in [0,\infty)^N\setminus \{0\}^N$ and denote the corresponding normalised weights by $w^{\ind{i}} = g^{\ind{i}}/\sum_{j=1}^N g^{\ind{j}}$, and the cumulative distribution function by $F(0)\equiv 0$ and $F(i) = \sum_{j=1}^i w^{\ind{j}}$ for $i\in [N]$. The generalised inverse $F^{-1}(u)$ is defined for $u\in(0,1)$ as the unique index $i\in[N]$ such that $F(i-1) < u \le F(i)$.

\begin{definition}[Systematic resampling]
  \label{def:systematic}
  Simulate a single $ \sim U(0,1)$, set 
  \[
     \check{U}^{\ind{i}} \defeq \frac{i-1 + U}{N}
  \]
  and define the resampling indices as $A^{\ind{i}} \defeq F^{-1}(\check{U}^{\ind{i}})$ for $i\in[N]$.
\end{definition}

\begin{definition}[Stratified resampling]
  \label{def:stratified}
  Simulate $U^{\ind{1:N}} \sim U(0,1)$, set
  \[
     \check{U}^{i} \defeq \frac{i-1 + U^{\ind{i}}}{N},
  \] 
  and define the resampling indices as $A^{\ind{i}} \defeq F^{-1}(\check{U}^{\ind{i}})$ for $i\in[N]$.
\end{definition}

We consider slightly modified versions of these resampling schemes, which rely
on an auxiliary ordering of weights. This allows for simpler analysis, but our
experiments also suggest potential performance gains.

\begin{definition}[Mean partition]
  \label{def:mean-partition}
  Suppose that $u^{\ind{1:N}}\in\R^N$. A permutation $\varpi:[N]\to[N]$ is a \emph{mean partition} (order) for $u^{\ind{1:N}}$, if the re-indexed vector $u_\varpi^{\ind{i}}\defeq u^{\ind{\varpi(i)}}$ satisfies $u_\varpi^{\ind{1}},\ldots,u_\varpi^{\ind{m}} \le \bar{u}$ and $u_\varpi^{\ind{m+1}},\ldots, {u}_\varpi^{\ind{N}} > \bar{u}$ for some $m\in[N]$, where $\bar{u} = N^{-1}\sum_{i=1}^N u^{\ind{i}}$.
\end{definition}

A mean partition $\varpi$ can be found in $O(N)$ time using Hoare's scheme \cite{hoare}.

\begin{definition}[Systematic/stratified resampling with order $\varpi$]
  \label{def:with-mean-partition}
  Let $F_\varpi^{-1}$ denote the generalised inverse distribution function corresponding to the re-indexed weights $w_\varpi^{\ind{1:N}}$.
Set
\[
  A^{\ind{\varpi(i)}} \defeq  \varpi\big(F_\varpi^{-1}(\check{U}^{\ind{i}})\big),
\]
where $\check{U}^{\ind{1:N}}$ are defined as in systematic/stratified resampling.
\end{definition}

In words, Definition \ref{def:with-mean-partition} means that we process the particles in order $\varpi$ within systematic/stratified resampling. We obtain the following convergence results, whose proofs are given in Appendix \ref{app:resampling-proofs}.

\begin{proposition}
  \label{prop:stratified-generator}
  Let $\varpi$ be a mean partition for $-v^{\ind{1:N}}$.
Stratified resampling with order $\varpi$ (Definition \ref{def:with-mean-partition}) satisfies Assumption \ref{a:resampling-generator} with resampling intensity
$$
\iota_{\mathrm{stratified}}(a^{\ind{1:N}},v^{\ind{1:N}})
= \begin{cases} \sum_{j=1}^i \big(v^{\ind{\varpi(j)}} - \bar{v}\big), 
& a^{\ind{\varpi(i)}} = \varpi(i+1),\, a^{\ind{\varpi(\neg i)}} = \varpi(\neg i)\text{ for }i\in[N-1]. \\
0, &\text{otherwise}
\end{cases}
$$
The overall resampling rate is
$$
\iota^*_{\mathrm{stratified}}(v^{\ind{1:N}})
= \sum_{j=1}^{N} j (\bar{v} - v^{\ind{\varpi(j)}}).
$$
\end{proposition}

\begin{proposition}
  \label{prop:systematic-generator}
  Let $\varpi$ be a mean partition for $-v^{\ind{1:N}}$.
Systematic resampling with order $\varpi$ (Definition \ref{def:with-mean-partition}) satisfies Assumption \ref{a:resampling-generator} with resampling intensity
$$
\iota_{\mathrm{systematic}}(a^{\ind{1:N}},v^{\ind{1:N}})
= \begin{cases} \big(\min\{ s_\varpi^{\ind{k}}, s_\varpi^{\ind{\ell-1}} \} - \max\{ s_\varpi^{\ind{k-1}}, s_\varpi^{\ind{\ell}}\} \big)_+, 
& a^{\ind{\varpi(1:N)}} = \varpi([k\to \ell]_N), \\
0, &\text{otherwise},
\end{cases}
$$
where $k,\ell$ are such that $v^{\ind{\varpi(k)}} \ge \bar{v}$ and $v^{\ind{\varpi(\ell)}} < \bar{v}$ and $s_\varpi^{\ind{0}} \defeq 0$,  $s_\varpi^{\ind{i}} \defeq \sum_{j=1}^i (v^{\ind{\varpi(j)}} - \bar{v})$. 
The overall resampling rate is
$$
\iota^*_{\mathrm{systematic}}(v^{\ind{1:N}})
= \sum_{i=1}^{N} (\bar{v} - v^{\ind{i}})_+.
$$
\end{proposition}

\begin{remark}
In a discrete time implementation, $\varpi$ can be selected as mean partition
of unnormalised weights $g^{\ind{1:N}}$. In practice, when the mean partition
$\varpi$ for $g^{\ind{1:N}} = \exp(-\Delta v^{\ind{1:N}})$ is computed by Hoare's scheme, the mean partition will converge to a mean partition for $-v^{\ind{1:N}}$.
\end{remark}

\begin{remark}
  We believe that 'plain' stratified and systematic resampling (without mean
  partition) also satisfy Assumption \ref{a:resampling-generator}, but verification
  becomes more technical. In addition, our empirical findings suggest that the mean
  partitioned versions can be preferable.
\end{remark}

\subsection{SSP resampling}

We consider next a variant of SSP resampling \citep{gerber-chopin-whiteley}
based on a processing order (permutation) $\varpi$; see Algorithm \ref{alg:ssp}.
\begin{algorithm}
  \caption{\textsc{SSPResampling}$(w^{\ind{1:N}},\varpi)$}
  \label{alg:ssp}
  \begin{algorithmic}[1]
  \State Let $r^{\ind{1:N}} \gets \lfloor N w^{\ind{1:N}} \rfloor$, $p^{\ind{1:N}} \gets N w^{\ind{1:N}} - r^{\ind{1:N}}$. and $(i,j)\gets(\varpi(1),\varpi(2))$.
  \For{$k=2,\ldots,N$}
  \State \label{line:set_deltas} Set $\delta^{\ind{i}} \gets \min\{ p^{\ind{j}}, 1 - p^{\ind{i}}\}$ and 
  $\delta^{\ind{j}} \gets \min\{ p^{\ind{i}}, 1 - p^{\ind{j}}\}$ 
  \State \label{line:interchange_indices} With probability $\charfun{\delta^{\ind{i}}>0}\frac{\delta^{\ind{i}}}{\delta^{\ind{i}}+\delta^{\ind{j}}}$, interchange $(i,j) \gets (j,i)$.
  \If{$p^{\ind{i}} + p^{\ind{j}} < 1$}
  \State Set $p^{\ind{i}} \gets p^{\ind{i}} + \delta^{\ind{i}}$ and $j \gets \varpi(\min\{k+1,N\})$.
  \Else
  \State Increment $r^{\ind{i}} \gets r^{\ind{i}} + 1$, 
  set $p^{\ind{j}} \gets p^{\ind{j}} - \delta^{\ind{i}}$ and $i \gets \varpi(\min\{k+1,N\})$.
  \EndIf
  \EndFor
  \State Return $A^{1:N} \gets $\textsc{RepeatIndices}$(r^{\ind{1:N}})$
  \end{algorithmic}
\end{algorithm}
The function \textsc{RepeatIndices}$(r^{\ind{1:N}})$ in Algorithm \ref{alg:ssp}
returns the non-decreasing index vector $A^{1:N}$ such that $\#\{j\in[N] :  A^{j} = i\} = r^{\ind{i}}$.

\begin{proposition}
  \label{prop:ssp-intensity}
SSP resampling with mean partition order $\varpi$ of $-g^{\ind{1:N}}$ satisfies
Assumption \ref{a:resampling-generator} with intensity 
$$\iota_{\mathrm{ssp}}(a^{\ind{1:N}},v^{\ind{1:N}})
= \begin{cases} 
    \displaystyle \frac{(v^{\ind{k}} - \bar{v} )_+(\bar{v} - v^{\ind{\ell}} )_+}{\iota_{\mathrm{ssp}}^*(v^{\ind{1:N}})}
& a^{\ind{1:N}} = [k\to \ell]_N \\
0, &\text{otherwise},
\end{cases}
$$
with overall resampling intensity $\iota^*_{\mathrm{ssp}}(v^{\ind{1:N}}) = 
\iota^*_{\mathrm{systematic}}(v^{\ind{1:N}}) = \sum_{i=1}^N (v^{\ind{i}} - \bar{v})_+$.
\end{proposition}

Proposition \ref{prop:ssp-intensity} follows directly from Proposition \ref{prop:ssp-probabilities} and Lemma \ref{lem:f-k-asymptotic-weights} in Appendix \ref{app:resampling-proofs}. 

\begin{remark}
  The overall resampling intensity of the SSP resampling coincides with systematic resampling with mean partition: $\iota^*_{\mathrm{ssp}}(v^{\ind{1:N}}) = \iota^*_{\mathrm{systematic}}(v^{\ind{1:N}})$. A closer inspection reveals that also the \emph{marginal intensities} for elimination of a particle $k$, or duplication of particle $\ell$, coincide. However, the elimination and duplication indices $K,L$, respectively, are independent in the case of SSP resampling, in contrast with systematic resampling, where they have a (somewhat complicated) dependence.
\end{remark}
  
\subsection{Comparison of resampling rates and a simplified limiting scheme}

In killing, and in stratified/systematic/SSP resampling based on a mean partition $\varpi$, exactly one particle is eliminated and one is duplicated in the limit. 
Therefore, the overall resampling event rate $\iota^*$ determines the instantaneous expected number of 'deaths' in all of these schemes. This motivates comparing the overall resampling rates.

\begin{theorem}
  \label{thm:resampling-order}
The overall resampling intensities of killing, stratified and systematic/SSP resampling with mean partition $\varpi$ of $-v^{\ind{1:N}}$ satisfy
\[
   \iota_{\mathrm{killing}}^*(v^{\ind{1:N}}) \ge \iota_{\mathrm{systematic}}^*(v^{\ind{1:N}}),\qquad\text{and}\qquad
   \iota_{\mathrm{stratified}}^*(v^{\ind{1:N}}) \ge \iota_{\mathrm{systematic}}^*(v^{\ind{1:N}}),
\]
for all potential values $v^{\ind{1:N}}$. However, 
$\iota_{\mathrm{killing}}^*$ and $\iota_{\mathrm{stratified}}^*$ do not satisfy such order in general.
\end{theorem}

Theorem \ref{thm:resampling-order}, whose proof is given in Appendix \ref{app:resampling-proofs}, shows that systematic and SSP resampling have the smallest overall resampling rate among the studied algorithms, which suggests that they may therefore be preferable over killing and stratified resampling. 

Let us conclude this section with another scheme, which has the same limit as SSP resampling, but with more a transparent behaviour. Note that this scheme can only be used with fine enough discretisations.
\begin{definition}[Symmetrised systematic resampling]
    \label{def:symmetrised-systematic}
Assume that $g^{\ind{1:N}}$ are such that their corresponding normalised weights $w^{\ind{1:N}}$ satisfy $p \defeq \sum_{i=1}^N (N w^{\ind{i}} - 1)_+ \le 1$ (cf.~Assumption \ref{a:nearly-constant-partitioned-weights} in Appendix \ref{app:resampling-proofs}).

With probability $1-p$, return $A^{\ind{1:N}} = 1{:}N$; otherwise pick indices $K$ and $L$ on $[N]$ independently with probabilities 
$$
\P(K=k) = \frac{(1 - N w^{\ind{k}})_+}{p}\qquad\text{and}\qquad 
\P(L=\ell) = \frac{(N w^{\ind{\ell}} - 1)_+}{p},
$$
and return $A^{\ind{1:N}} = [k\to \ell]_N$.
\end{definition}

The following proposition is straightforward to check given Lemma \ref{lem:f-k-asymptotic-weights} in Appendix \ref{app:resampling-proofs}.
\begin{proposition}
    \label{prop:symmetrised-systematic-intensity}
Symmetrised systematic resampling (Definition \ref{def:symmetrised-systematic}) satisfies Assumption \ref{a:resampling-generator} with intensity
$ \iota_{\mathrm{s.syst}}(a^{\ind{1:N}},v^{\ind{1:N}})
= \iota_{\mathrm{ssp}}(a^{\ind{1:N}},v^{\ind{1:N}})$.
\end{proposition}

\begin{remark}
  Symmetrised systematic resampling algorithm (Definition \ref{def:symmetrised-systematic}) can be used in place of another resampling scheme, such as SSP resampling, whenever the required condition is met (i.e.~$p\le 1$). Such a combination would yield slight computational benefits, as the symmetric systematic resampling only requires two uniform random variables.
\end{remark}

\section{Convergence to a continuous-time limit} 
\label{sec:convergence}

Here we present a convergence result for particle filters as in Algorithm \ref{alg:pf}, targeting a time-discretised path-integral model as discussed in Section \ref{sec:path-integral}. The state space is $\X \defeq \real^d$ and the transitions $M_k$ correspond to appropriately scaled Euler-Maruyama type discretisations of the $d$-dimensional diffusion
\begin{equation}
  \label{eq:sde}
  \ud z_t = b(z_t)\ud t + \sigma(z_t) \ud W_t,
\end{equation}
with $z_0 \sim \mu$ for some fixed $\mu \in \mathcal{P}(\real^d)$ and coefficient functions $b\colon \real^d\to\real^d$ and $\sigma\colon \real^d\to\real^d\times\real^d$ specified below.

To this end, let $\tau > 0$ be a continuous-time horizon, $V\colon \real^d \to [0,\infty)$ be a bounded and continuous potential function and $(\Delta_n)_{n \in \N} \subset (0,\tau \land 1)$ be an arbitrary decreasing sequence of discretisation step sizes converging to zero.

For $n \in \N$, write $\widetilde{X}^{\Delta_n}$ for the $(\real^{d})^N$-valued Markov chain given by Algorithm \ref{alg:pf} with $T = \lfloor \tau / \Delta_n \rfloor + 1$, resampling scheme $r$ satisfying Assumption \ref{a:resampling-generator} and the transitions $(M_k)_{k \in 1{:}T}$ and functions $(G_k)_{k\in 1{:}(T-1)}$ defined as follows:
\begin{itemize}
  \item $M_1 = \mu$, and $M_k(\ud y \mid x) = \P\big( x + b(x)\Delta_n + \sigma(x) B^{\Delta_n}_k \in \ud y \big)$ for $k \in 2{:}T$, where the $B^{\Delta_n}_k$'s are independently distributed as $\mathcal{N}(0,\Delta_n I_{\real^d})$;
  \item $G_{k}(\underline{\vec{x}}_{k}) \defeq \nu^{\Delta_n}(\underline{x}_k) \defeq \big( e^{-\Delta_n V(x_k^{\ind{1}})}, \cdots, e^{-\Delta_n V(x_k^{\ind{N}})} \big)$ for all $k \in 1{:}(T-1)$.
\end{itemize}
Note that each $G_k$ in the above definition depends only on the states of the particles at time $k$, so that $(\widetilde{X}^{\Delta_n}_k)_{k \in 1{:}T}$ is indeed a Markov chain.

Write then
\[
  X^{\Delta_n}_k \defeq \widetilde{X}^{\Delta_n}_{k+1} \quad \text{for} \quad k \in \{0,1,\cdots,\lfloor \tau / \Delta_n\rfloor\}.
\]
This re-indexing is introduced since particles commence with `time' index $1$ in Algorithm \ref{alg:pf}. The next theorem proves convergence of the c\`adl\`ag  extension of this skeleton which is defined as $X^{\Delta_n}_{\lfloor t/\Delta_n \rfloor}$ for $t\in[0,\tau]$.

Recall that by Assumption \ref{a:resampling-generator},
\[
  \lim_{n\to\infty} \frac1{\Delta_n} r\big( a \mid \nu^{\Delta_n}(x) \big) = \iota\big( a , (V(x^{\ind{1}}),\cdots,V(x^{\ind{N}})) \big) =: \iota^a (x)
\]
for all $a \in [N]^N\setminus\{1{:}N\}$, with \emph{bounded and pointwise convergence} with respect to $x \defeq (x^{\ind{1}},\cdots,x^{\ind{N}}) \in (\real^d)^N$ in the sense that the term inside the limit is uniformly bounded with respect to $n$ and $x$.

\begin{theorem}\label{thm:convergence}
Let the $(\real^d)^N$-valued Markov chains $X^{\Delta_n}$, $n \in \N$, be as
above. Assume that the coefficient functions $b$ and $\sigma$ of the diffusion
\eqref{eq:sde} are Lipschitz continuous and bounded, that the diffusion is
uniformly non-degenerate in the sense that
\begin{equation}\label{eq:nondegenerate}
  \inf_{x\in\real^d} \, \inf_{\theta \in \real^d,\,|\theta|=1} |\sigma(x) \theta| > 0,
\end{equation}
and that the functions $\iota^a\colon (\real^d)^N \to [0,\infty)$, $a \in [N]^N\setminus\{1{:}N\}$, are bounded and continuous. 
\begin{enumerate}[(i)]
\item There exists a continuous-time process $(Z_t)_{t \geq 0}$ with c\`adl\`ag paths in $(\R^d)^N$ such that
\[
  \lim_{n \to \infty} ( X^{\Delta_n}_{\lfloor t_1/\Delta_n \rfloor}, \cdots , X^{\Delta_n}_{\lfloor t_T/\Delta_n \rfloor} )
  = ( Z_{t_1}, \cdots , Z_{t_T} )
\]
in distribution for all finite $\{t_1,\cdots,t_T\} \subset [0,\tau]$.
\item The limit process $Z$ in part (i) has infinitesimal generator
\[
  \L f(x) \defeq \sum_{j \in [N]} L^{(j)}f (x) + \sum_{a \in [N]^N\setminus\{1{:}N\}} \iota^a(x) \big( f(x^{\ind{a(1{:}N)}}) - f(x) \big)
\]
for $f \in \testf( \real^{dN} )$ and $x \in \real^{dN}$, where $L$ is the generator corresponding to the $d$-dimensional diffusion \eqref{eq:sde}, $L^{(j)}f(x)$ stands for $L[ y \mapsto f(x^{\ind{1}},\cdots,x^{\ind{j-1}},y,x^{\ind{j+1}},\dots,x^{\ind{N}})](x^{\ind{j}})$ and $x^{\ind{a(1{:}N)}} = (x^{\ind{a(1)}},\cdots,x^{\ind{a(N)}}) \in \real^{dN}$ for $x \defeq (x^\ind{1},\cdots,x^{\ind{N}})\in \real^{dN}$ and $a \in [N]^N$.
\item Let $\mathcal{V}\colon[0,\infty)\times \R^{dN} \to [0,\infty)$ be a bounded and continuous function. Then
\[
  \lim_{n\to\infty} \E \Big[ f(X^{\Delta_n}_{\lfloor \tau/\Delta_n\rfloor} ) \prod_{k=0}^{\lfloor \tau/\Delta_n\rfloor-1} e^{- \Delta_n \mathcal{V}( k\Delta_n, X^{\Delta_n}_k )}\Big] = \E\Big[ f(Z_{\tau}) \exp\big( -\int_0^{\tau} \mathcal{V}(u,Z_u)\ud u \big)\Big]
\]
for all bounded and continuous $f\colon \R^{dN} \to \R$.
\end{enumerate}
\end{theorem}

The proof of Theorem \ref{thm:convergence} is given in Appendix \ref{app:convergence} below.

\begin{remark}
    Regarding Theorem \ref{thm:convergence}:
    \begin{enumerate}[(i)]
        \item The assumption about the boundedness and continuity of the
            functions $\iota^a$ above often follows automatically from the
            corresponding properties of the potential function $V$ (cf.
            Proposition \ref{pr:killing-resampling-rate}).
        \item Theorem \ref{thm:general-convergence} in Appendix
            \ref{app:convergence} is a more general variant of part (iii) of
            Theorem \ref{thm:convergence}.
        \item 
            The result is formulated for time-homogeneous coefficient and
            potential functions $b$, $\sigma$ and $V$ for simpler exposition.
            In fact, by considering a time-augmented state space
            (i.e.~$[0,\infty)\times\real^{dN}$ instead of $\real^{dN}$), an
            analysis similar to the one in Appendix \ref{app:convergence} can
            be carried out for time-dependent coefficient and potential
            functions, resulting in a variant of the Theorem with a
            time-inhomogeneous limit process $Z$. We omit the details; see
            e.g.~\cite[Chapter 4, Section 7]{ethier-kurtz} for basic results
            corresponding to such generalizations.
        \item In the special case where killing resampling is used, we recover
            as the limit the continuous-time particle system described in e.g.
            Section 1.5.2 of \cite{del-moral} or
            \cite{del-moral-jacob-lee-murray-peters}. See also e.g.~\cite[Example 3.1.3 and Proposition 3.4]{rousset} for a continuos-time particle model with overall resampling rate that interestingly coincides with that of the systematic and SSP resampling schemes (see Propositions \ref{prop:systematic-generator} and \ref{prop:ssp-intensity}). 
    \end{enumerate}
\end{remark}

\section{Unbiased estimation of Feynman-Kac measures}
\label{sec:unbiasedestimation}

Continuing the theme of Section \ref{sec:convergence}, we explain how the unbiasedness condition of the resampling $r$ (Assumption \ref{a:unbiased-resampling}) leads to an unbiasedness property for the jumping intensities $\iota(a,\cdot)$ (Definition \ref{def:asymptotic-unbiased} below). This property will be applied for time-marginal Feynman-Kac measures of the particle filter on the continuous-time limit, namely Theorem \ref{thm:fk-unbiased} below, which is a continuous-time variant of the well-known property \eqref{eq:unbiased-smoothing}. 

\begin{definition}\label{def:asymptotic-unbiased}
  We say that a resampling scheme $r$ satisfying Assumption \ref{a:resampling-generator} is \emph{asymptotically unbiased} if
  \begin{equation}
    \label{eq:asymptotic-unbiased}
    \sum_{a \in [N]^N\setminus\{1{:}N\}} \iota(a,v^{\ind{1{:}N}})\big(\#\{j \in [N]\,:\, a^\ind{j} = i\} - 1\big) = 
    \frac{1}{N} \sum_{j=1}^N v^{\ind{j}} - v^{\ind{i}}
  \end{equation}
for all $v^{\ind{1{:}N}} \in [0,\infty)^N$ and $i \in [N]$.
\end{definition}

In order to state the main result of this section, let us introduce the following notation: for functions $f \colon \real^d\to\real$, write $\overline{f} \colon (\real^{d})^N \to \real$ for the function $x \mapsto \frac1N \sum_{i=1}^N f(x^{\ind{i}})$.

\begin{theorem}\label{thm:fk-unbiased}
Let
\begin{enumerate}[(i)]
\item $r$ be an asymptotically unbiased resampling scheme;
\item $Z \defeq (Z_t)_{t \geq 0}$ be the continuous-time particle filter in Theorem \ref{thm:convergence};
\item $z \defeq (z_t)_{t \geq 0}$ be the solution to the $d$-dimensional stochastic differential equation \eqref{eq:sde} with initial distribution $\mu$ (the solution is unique in law under the assumptions of Theorem \ref{thm:convergence});
\item $V\colon \real^d \to [0,\infty)$ be the bounded and continuous potential function in (the construction preceding) Theorem \ref{thm:convergence}.
\end{enumerate}
Then
\begin{equation}\label{eq:fk-unbiased}
  \E\Big[ \overline{f}(Z_t) \exp\Big( - \int_0^t \overline{V}(Z_u) \ud u \Big)\Big] = \E\Big[ f(z_t) \exp \Big( - \int_0^t V(z_u)\ud u \Big)\Big]
\end{equation}
for all $t > 0$ and bounded and measurable $f \colon \real^d\to\real$.
\end{theorem}

Regarding assumption (i) in Theorem \ref{thm:fk-unbiased}, we note that it follows from our standard Assumptions \ref{a:unbiased-resampling} and \ref{a:resampling-generator}:

\begin{proposition}\label{pr:asymptotic-unbiased}
  Suppose that the resampling $r$ satisfies Assumptions \ref{a:unbiased-resampling} and \ref{a:resampling-generator}. Then it is asymptotically unbiased.
\end{proposition}

In particular, the assumption holds for all resampling schemes examined in Section \ref{sec:resampling-schemes}, but it may also hold for a richer class of resampling schemes, including ones that are not necessarily unbiased in the sense of Assumption \ref{a:unbiased-resampling}.

The proofs of Theorem \ref{thm:fk-unbiased} and Proposition \ref{pr:asymptotic-unbiased} are presented in Appendix \ref{app:fk-unbiased}.

\section{Experiments} \label{sec:experiments}

We compare empirically the behaviour of a number of resampling algorithms in two
experiments. Our first experiment is on an Ornstein-Uhlenbeck latent process and
a `box-shape' potential (`OU'). The second experiment is about the inference of
a Cox process, that is, an inhomogeneous Poisson process with latent intensity, where we use the particle filter within a PMMH (particle marginal Metropolis-Hastings) sampler. 

\subsection{Ornstein-Uhlenbeck process and box potential}
\label{sec:ou-example}

In this experiment, $\mathbb{M}$ corresponds to the law of a stationary
Ornstein-Uhlenbeck process $(Z_t)_{t\in [0,\tau]}$ with initial distribution
$Z_0\sim N(0,\sigma_\infty^2)$ which solves the following stochastic differential
equation
$$
\ud Z_t = - \theta Z_t \ud t + \sigma \ud W_t,
$$
where $(W_t)_{t\ge 0}$ is the standard Brownian motion. The parameters are
$\theta=0.1$ and $\sigma=1$ and the stationary variance is $\sigma_\infty =
\sigma/\sqrt{2\theta} \approx 2.236$.

The potential function is a (relatively narrow) `box-shaped' potential function
of the following form:
$$
V(x) \defeq 6\times\charfun{| x - 0.5 | > 0.1}.
$$

We study the performance of the particle filter with different resampling
schemes in gradually finer discretisation $\log_2 \Delta \in \{-12, -10, \ldots, 0\}$, with different number of particles $N\in \{64, 128, 256, 512\}$. We
repeat the particle filter 10,000 times with each configuration to obtain the root mean squared errors (RMSE) in the figures. We consider the
resampling schemes discussed in Section \ref{sec:resampling-schemes}. The
resampling schemes that include mean partition order, are named with
`Partition'.

For each $\Delta$, we calculate the `ground truth' of the normalising constant, defined as $\mathcal{Z}_T$ in \eqref{eq:discretised-path-integral}, over all scenarios (all resampling schemes, all $N$). Taking this normalising constant as the truth, we construct unbiased estimators of the relative normalising constant (true value 1) and the filtering (the last state $\hat{Z}_\tau$ in \eqref{eq:discretised-path-integral}) and smoothing expectations (the first state $\hat{Z}_0$). For the latter, we use the `filter smoother', that is, use traced back paths in estimation. 

\begin{figure}
    \includegraphics[width=.99\linewidth]{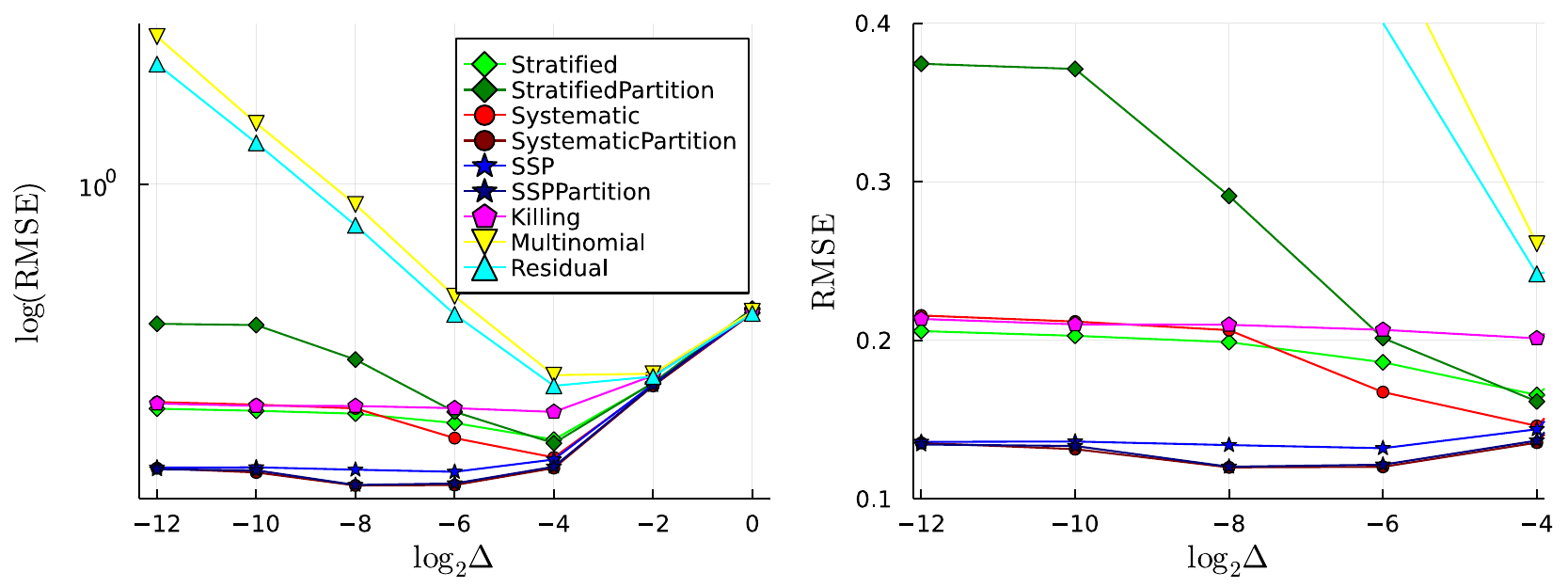}
    \caption{Normalising constant estimates with $N=512$ in the OU example.
        While
        the left panel uses a log scale and compares all the resampling
        schemes, the right panel uses a linear scale to compare more finely
        those resampling schemes that do not diverge as $\Delta\rightarrow 0 $.
    } \label{fig:ou-norm-const-512}
\end{figure}
Figure \ref{fig:ou-norm-const-512} shows the normalising constant estimate mean
squared errors (MSEs) in the case $N=512$. When $\Delta = 2^{0}$, the
performance is almost identical across resampling schemes, as we are not in the
weakly informative setting, but this is no longer the case as $\Delta \to 0$.
As expected, the performance of multinomial and residual resampling decay as
$\Delta\to 0$, while other resampling schemes remain stable.  The increase in
relative RMSE for the multinomial scheme supports related results in
\cite{cerou-delmoral-guyader} that show  the  variance of the normalising
constant estimate increases exponentially with the number of resampling
instances. Similar to multinomial, residual resampling rate does not stabilise
(see comment after Assumption \ref{a:resampling-generator}) and hence the
observed variance increase as $\Delta$ decreases.  The zoomed figure on the
right suggests that the best performance is obtained with SystematicPartition
and SSPPartition. SSP resampling is also close, and seems indistinguishable in
the $\Delta\to 0$ limit.

\begin{figure}
  \includegraphics[width=.99\linewidth]{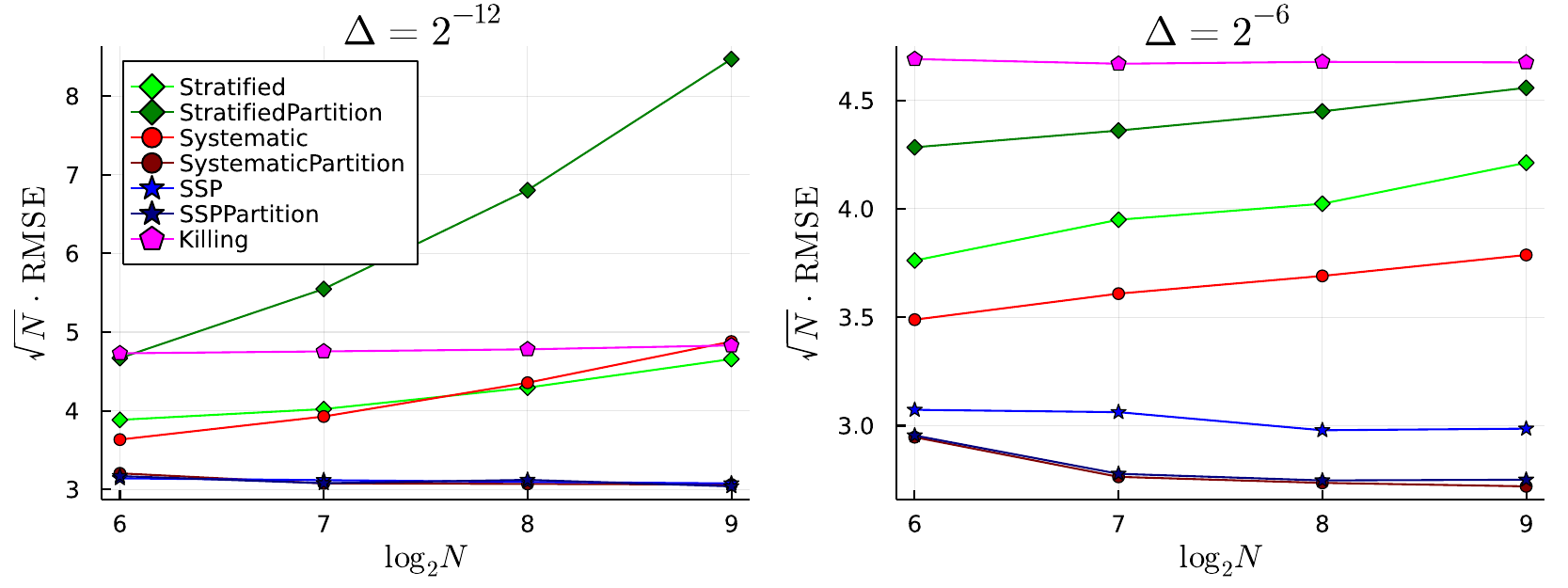}
  \caption{Normalising constant estimates with different $N$ in the OU example.
    Each panel corresponds to a different value of $\Delta$.}
  \label{fig:ou-norm-const-varying-N}
\end{figure}
Figure \ref{fig:ou-norm-const-varying-N} shows a similar picture for varying
$N$, with two choices of $\Delta$. The values of the y-axis are RMSEs multiplied
by $\sqrt{N}$, which is expected to stabilise if a central limit theorem type
result holds. The results suggest that Killing, SSP, SSPPartition and
SystematicPartition indeed stabilise, with the latter two being again the best.

\begin{figure}
  \includegraphics[width=.99\linewidth]{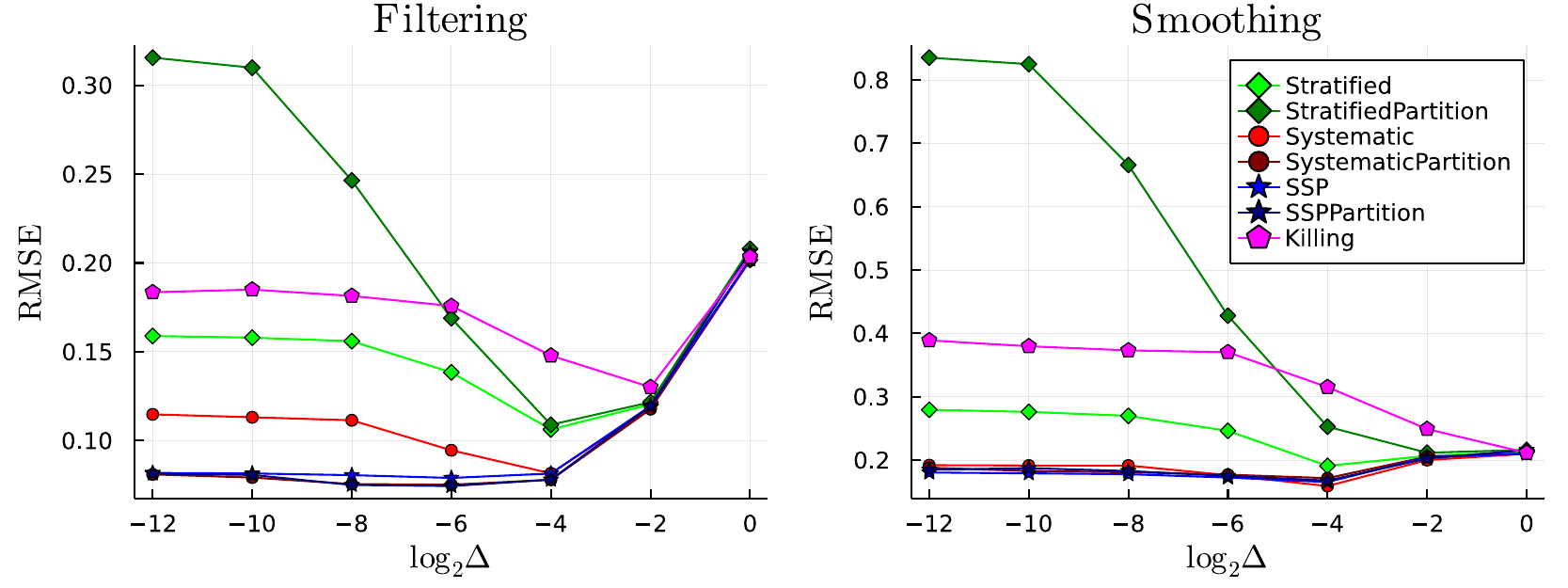}
  \caption{Filtering and smoothing estimates with $N=512$ in the OU example.}
  \label{fig:ou-filtering-smoothing-512}
\end{figure}
Figure \ref{fig:ou-filtering-smoothing-512} shows filtering and smoothing
estimate MSEs for stable resampling schemes similar to Figure
\ref{fig:ou-norm-const-512}. The conclusions are similar, except for systematic
resampling, which seems to be competitive with the best schemes in smoothing,
but not in filtering.

\begin{figure}
  \includegraphics[width=.99\linewidth]{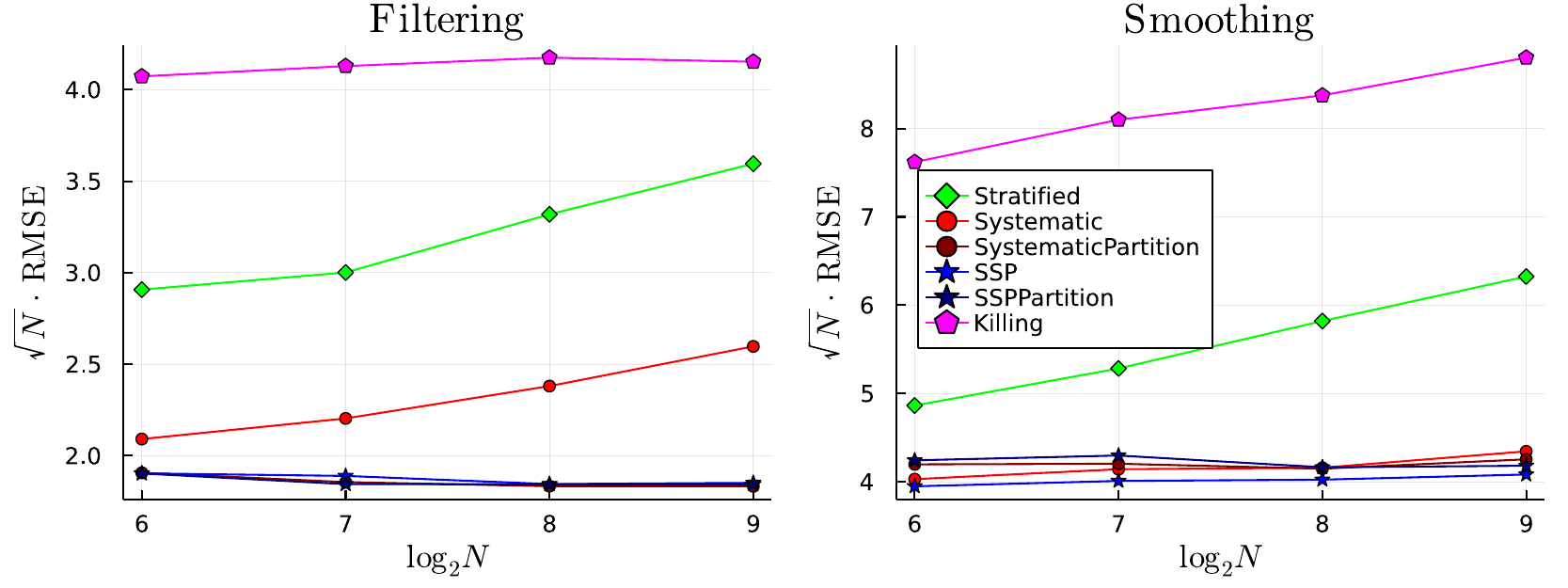}
  \caption{Filtering and smoothing estimates with $\Delta = 2^{-12}$ and varying $N$ in the OU example.}
  \label{fig:ou-filtering-smoothing-varying-N}
\end{figure}
Figure \ref{fig:ou-filtering-smoothing-varying-N} shows filtering and smoothing estimate MSEs scaled with $\sqrt{N}$ similar to Figure \ref{fig:ou-norm-const-varying-N}. Again, Killing, SSP, SSPPartition and SystematicPartition seem stable in the case of filtering, but smoothing with Killing has not stabilised yet.

\subsection{Comparison with adaptive resampling}

Adaptive resampling \cite{liu-chen-blind} is a commonly used method with particle filters, where resampling is performed only when so-called effective sample size of the weights falls below a predefined threshold (fraction of particles). Adaptive resampling is out of the scope of our theoretical framework, but can be useful in practice also in the weak potentials setting, so we compare empirically how adaptive resampling performs in the OU example (Section \ref{sec:ou-example}).

Figure \ref{fig:ou-adaptive-stability} shows the performance with adaptive resampling with $N=512$ particles and threshold $t_{\mathrm{res}}=0.5$, in the filtering and smoothing, similar to Figure \ref{fig:ou-filtering-smoothing-512}. 
\begin{figure}
  \includegraphics[width=.99\linewidth]{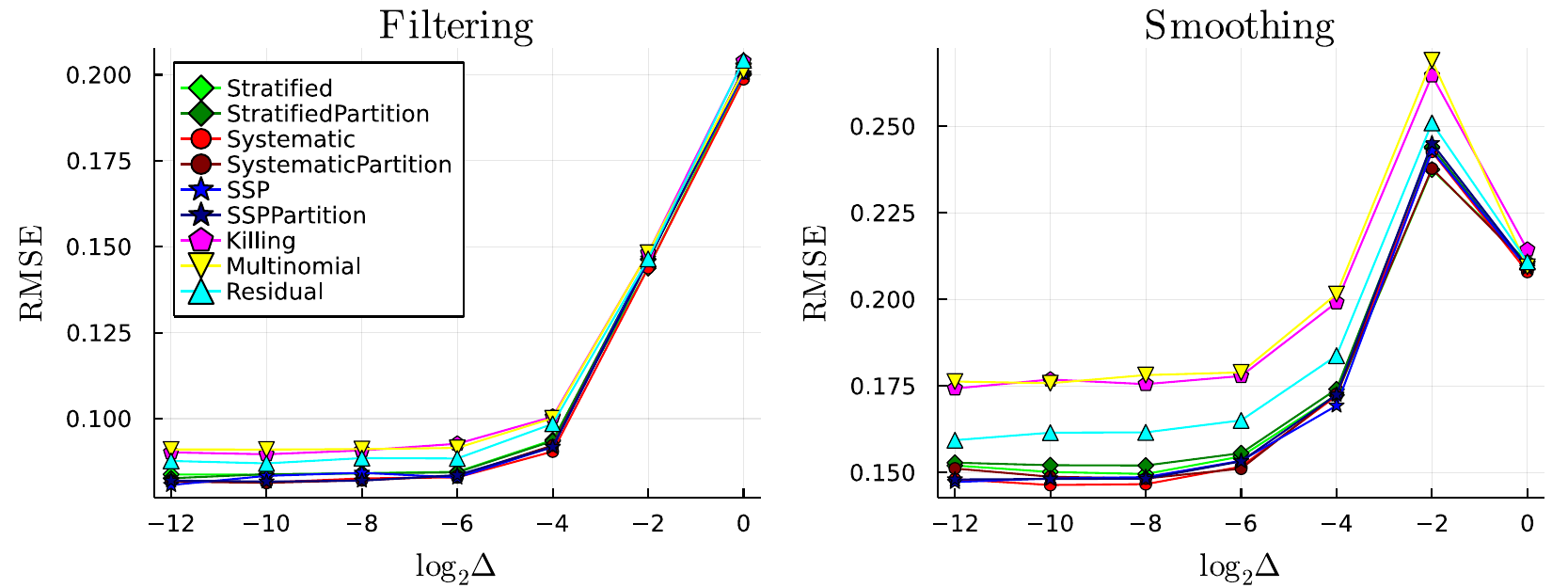}
  \caption{Adaptive resampling with threshold $0.5$ and $N=512$ in the OU example.}
  \label{fig:ou-adaptive-stability}
\end{figure}
Adaptive resampling seems stable with all resamplings, the differences between resamplings are small, and the performance is competitive with the best non-adaptive resamplings. The behaviour with smaller number of particles is qualitatively similar to $N=512$ (results not shown).

The threshold value, which controls how often resampling is triggered, is a tuning parameter of the method. We repeated the experiment with a range of thresholds $t_{\mathrm{res}}\in\{k/8\given k=0,\ldots,8\}\cup \{1-2^{-k}\given k=4,\ldots,9\}$. Figure \ref{fig:ou-adaptive-thresholds} shows a comparison of the results with finest discretisation $\Delta= 2^{-12}$.
\begin{figure}
  \includegraphics[width=.99\linewidth]{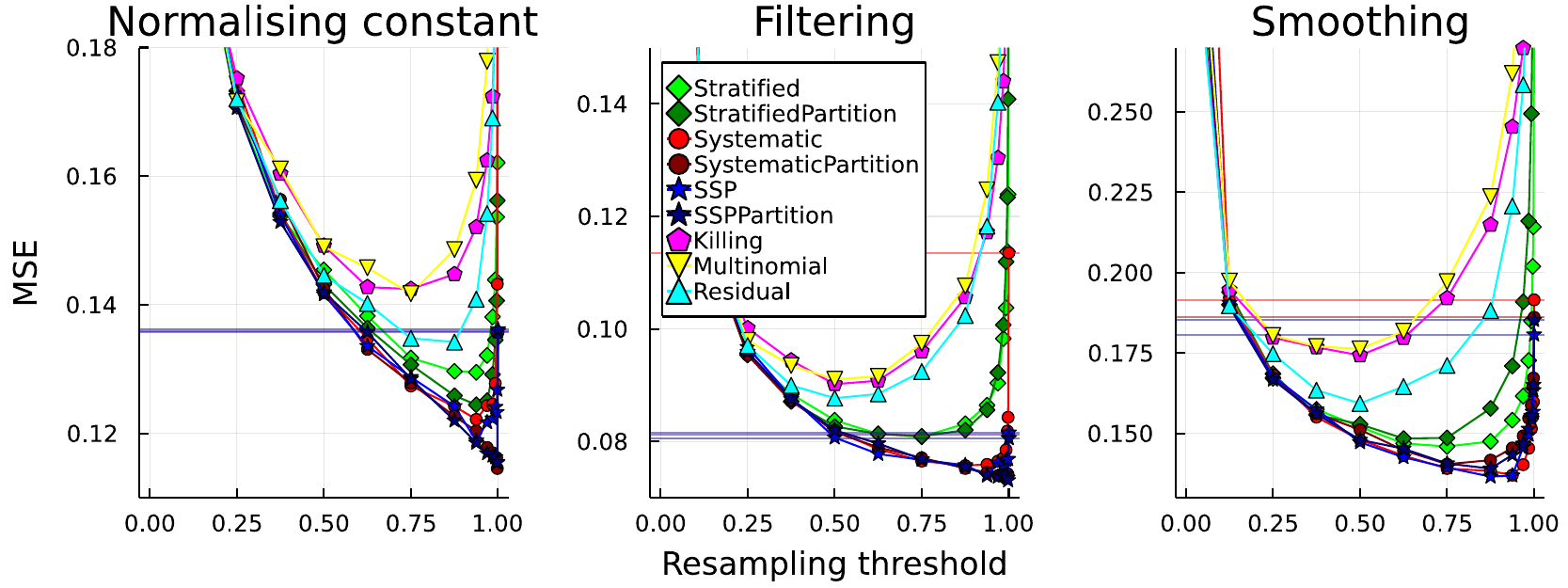}
  \caption{Adaptive resampling with $\log_2\Delta = -12$ and varying threshold in the OU example. The horizontal lines indicate the performance with non-adaptive resmplings ($t_{\mathrm{res}}=1.0$).}
  \label{fig:ou-adaptive-thresholds}
\end{figure}
The differences between resamplings are small with low threshold values, but more noticeable with higher thresholds.
For normalising constant estimation and filtering, adaptive Multinomial resampling does not reach the efficiency of the best non-adaptive schemes. In contrast, adaptive resampling can improve on the smoothing performance, and for instance with $t_{\mathrm{res}}=0.5$ all adaptive resamplings outperform the best non-adaptive resampling.
Interestingly, the optimal threshold value appears to depend on the resampling. For Multinomial, Residual and Killing, the optimal value is close to $0.5$, but for SSPPartition and SystematicPartition, the optimal threshold is closer to one. 

\subsection{Cox process with Particle marginal Metropolis-Hastings}\label{sec:Cox}

In our second example, we consider a Cox process model, that is, an inhomogeneous Poisson process with random intensity. We infer the latent intensity based on event times $0<\tau_1<\cdots<\tau_m<\tau$, leading to the following model:
\begin{equation}
 \Pi(\phi) \defeq \frac{1}{\mathcal{Z}_\M} \E_{\M}\bigg[ \phi\big(Z_{[0,\tau]}\big)
    \exp\bigg(-\int_0^\tau V_u(Z_u) \ud u\bigg) 
    \prod_{i=1}^m V_{\tau_i}(Z_{\tau_i})
    \bigg],
  \label{eq:path-integral-2}
\end{equation}
where $\mathbb{M}$ stands for the law of a reflected
Brownian motion on $[a,b]$, and the potential $V_u(z) = \beta e^{-\alpha z}$.

We approximate the reflected Brownian motion with the discrete-time dynamics $M_1 =
N(0,1)$ and $M_k(\cdot\mid X_{k-1})$ for $k\ge 1$:
$$
    \hat{X}_k \sim N(X_{k-1}, \Delta_k \sigma^2); \qquad X_k = \mathrm{reflect}(\hat{X}_k; a, b),
$$
where $\Delta_k = t_k - t_{k-1}$ is the time difference between $X_{k-1}$, and $X_k$ and $\mathrm{reflect}$ implements a folding back to $[a,b]$.

We consider synthetic data $\tau=(\tau_1,\ldots,\tau_n)$ generated from the model, where $Z_t$ is the c\`adl\`ag extension of the skeleton $Z_{t_k}=X_k$ on $[0,T]$. We use the constant step size $\Delta_k=\Delta=0.01$ and the parameter values $\sigma=0.3$, $\alpha=1.0$, $\beta =0.5$ and $T=200$ in the simulation.

We then use the particle marginal Metropolis-Hastings (PMMH)
\cite{andrieu-doucet-holenstein} to do posterior inference with independent
$N(0,2.5)$ prior for all log-transformed parameters $\theta = (\log \sigma, \log
\alpha, \log \beta )$. The discretisation mesh is a uniform grid as in the data
generation, augmented with the data points. The potentials are defined as
follows:
$$
   \log G_k(x) = - \Delta_{k+1} \beta \exp(-\alpha X_k) + 1(t_k \in \tau) (\log(\beta ) - \alpha X_k)
$$
That is, the latter part is only included in case of data point is observed at
$t_k$. The initial value of the PMMH is set to $\theta_0 = (0,0,0)$. We use the
continuous covariance adaptation scheme of \cite{haario-saksman-tamminen} within
PMMH \cite[cf.][]{vihola-amcmc} during the entire simulation of 500,000
iterations, with 50,000 taken as as burn-in. We repeat the experiment with
$N\in\{16,32,64,128,256\}$ particles and the same range of resampling algorithms
as in the previous experiment.

\begin{figure}
  \includegraphics[width=.99\linewidth]{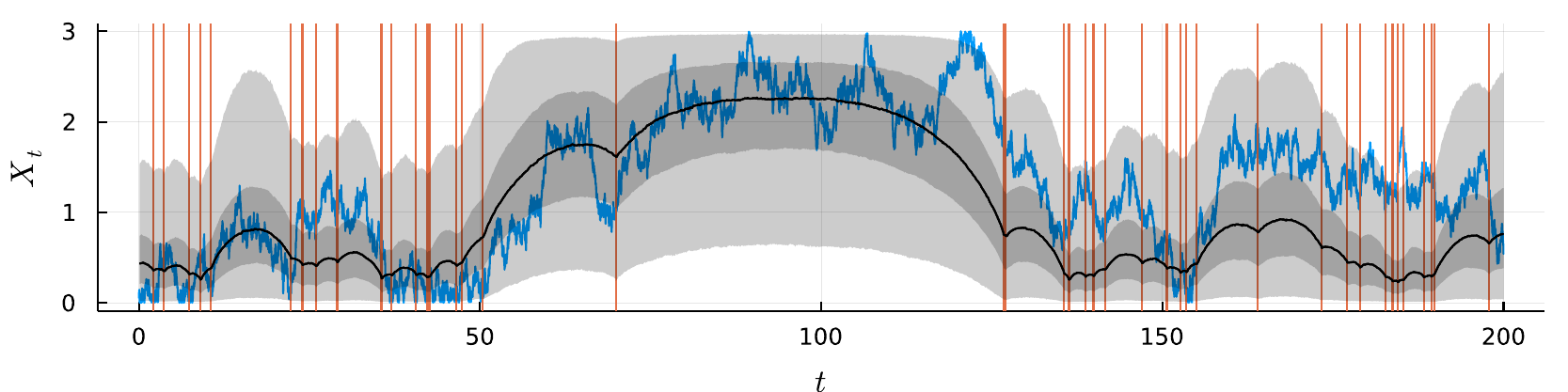}
  \caption{Generated latent state (blue) and observed times $y$ (red) in the Cox process experiment, as well as the posterior smoothing 50\% and 90\% credible intervals with PMMH using SystematicPartition and 32 particles.}
  \label{fig:luminosity-data-result}
\end{figure}
Figure \ref{fig:luminosity-data-result} shows the data in the experiment, and
illustrates the inference outcome for the latent state. It is intuitive that
there is substantial uncertainty in longer intervals with no observations. In
these intervals, the potentials are weak, and so the resampling strategy is
expected to have an impact in the efficiency.

\begin{figure}
  \includegraphics[width=.99\linewidth]{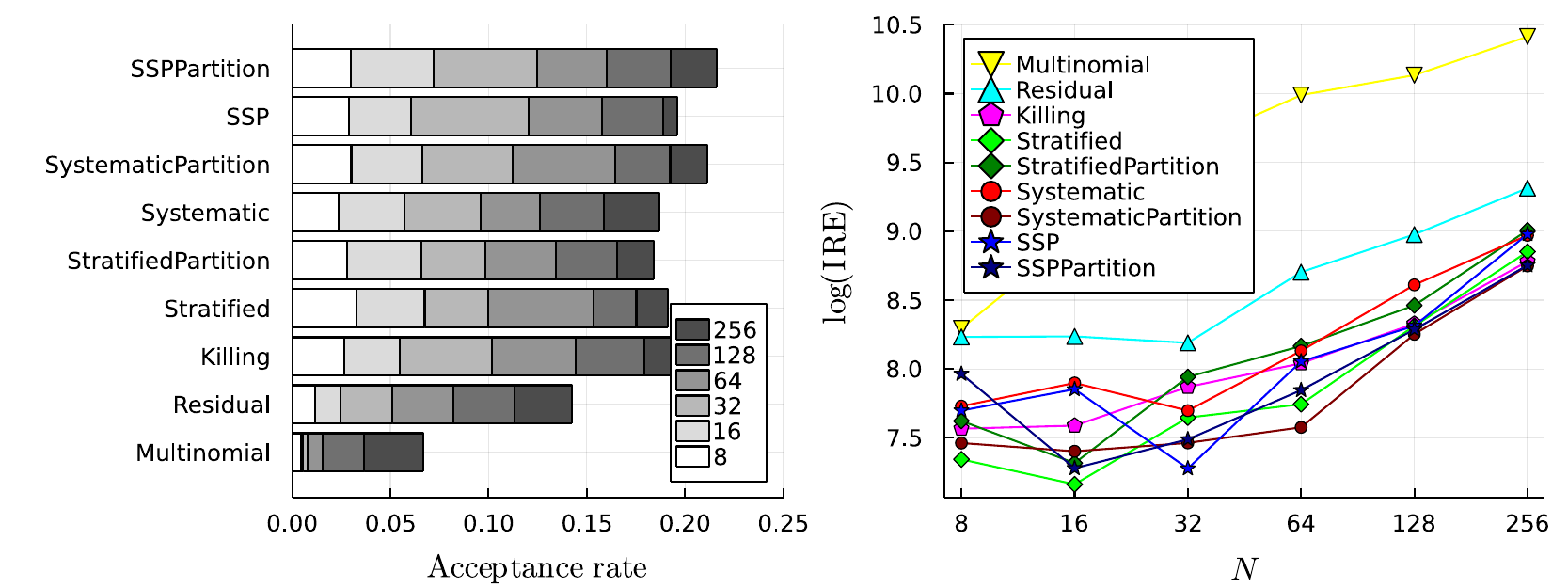}
  \caption{Acceptance rate (left) and mean inverse relative efficiency (right) of PMMH in the Cox process experiment.}
  \label{fig:luminosity-acc-ire}
\end{figure}
Figure \ref{fig:luminosity-acc-ire} (left) shows the PMMH acceptance rate in the
different scenarios. The same group of SSP, SSPPartition and SystematicPartition
attains the highest rates, and with multinomial and residual resampling, the
acceptance remains notably lower. To attain a 10\% acceptance rate, residual
resampling needs 128 particles in contrast with 32 particles for the best
resampling schemes.

Figure \ref{fig:luminosity-acc-ire} (right) illustrates the mean inverse
relative efficiencies (IREs) \citep{glynn-whitt}, that is, mean asymptotic
variances of the standardised log-transformed parameters, multiplied by number
of particles. The asymptotic variances are calculated by batch means
\cite{flegal-jones}, and standardisation is based on mean and variance estimates
calculated from all outputs. The results are in line with earlier findings, but
suggest that a low number of particles (even as low as 8) might be optimal in
some cases. However, this might well be anomaly due to underestimation of
asymptotic variance, which is supported by inspection of autocorrelation plots
of the first parameter ($\log \alpha$) shown in Figure
\ref{fig:luminosity-acf}.  Note that the lags are chosen inversely proportional to $N$ to account for varying cost per iteration.
\begin{figure}
  \includegraphics[width=.99\linewidth]{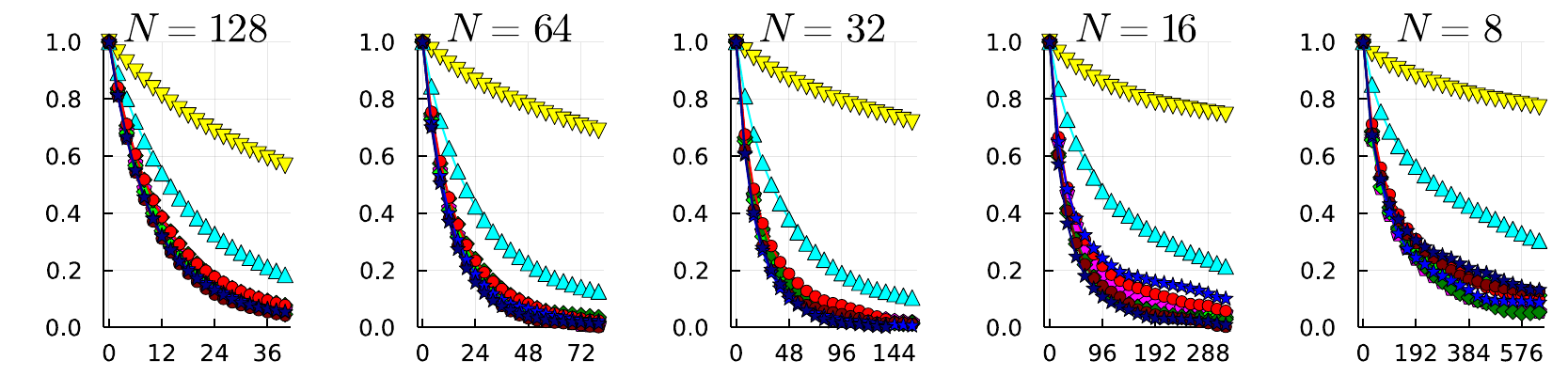}
  \caption{Autocorrelation functions of $\log \alpha$ in the Cox process experiment.}
  \label{fig:luminosity-acf}
\end{figure}

\section{Discussion}

We investigated the effect of resampling methods in a particle filter targeting
a HMM with uninformative observations, by considering discretisations of
continuous-time path integral models. We introduced a general condition for
discrete-time resampling schemes which guarantees convergence to a
non-degenerate particle system the continuous-time limit. We are unaware of
earlier results establishing continuous-time limits of particle filters with
different resampling strategies.

Resampling methods which satisfy our condition are `safe' to use with weakly
informative observations/potentials. We introduced modified versions of
stratified/systematic/SSP resampling, which are shown to satisfy the condition.
The modified strategies add a simple (and computationally cheap) algorithmic
step to the resampling schemes, which orders the weights about their mean value.
The modified algorithms lend themselves to a theoretical analysis, which reveals
that systematic and SSP resampling schemes yield the smallest overall resampling
rate, and may therefore be preferable.

Our empirical results complement our theoretical findings: systematic and SSP
resampling with mean ordering had the best performance in all experiments.
Because of the appealing theoretical properties of SSP resampling
\citep[cf.][]{gerber-chopin-whiteley}, it can be recommended also in the weakly
informative regime. However, the systematic resampling may remain preferable in
some settings, because of its slightly lower computational cost. Interestingly,
the mean partition order, which was necessary for theoretical analysis, appears
to improve the performance of systematic resampling as well. Based on our findings, we
recommend that systematic resampling is always used together with the mean
partition ordering of the weights. SSP resampling appears to perform well
also without such pre-ordering.

Adaptive resampling \citep{liu-chen-blind} and further refinements, such as partial interaction 
schemes \citep{whiteley-lee-heine}, can also be useful in the weakly informative setting, 
but are out of the scope of our theoretical framework. Our empirical comparison suggests that
adaptive resampling can further improve performance of the studied resampling algorithms. 
However, optimal choice of threshold is non-trivial, as it seems to depend on the resampling scheme.

\begin{appendix}

\section{Proofs for Section \ref{sec:resampling-schemes}}
\label{app:resampling-proofs} 

We first establish results for weights that are mean partitioned and nearly constant.
\begin{assumption}
  \label{a:nearly-constant-partitioned-weights}
  Let $w^{\ind{1:N}}$ be normalised weights and write $w^{\ind{i}} = \frac{1 + \epsilon^{\ind{i}}}{N}$ where $\epsilon^{\ind{i}} = Nw^{\ind{i}}-1$. Suppose that $\sum_i |\epsilon^{\ind{i}}| < 2$ and that there exists $m\in[N]$ such that $\epsilon^{\ind{1}},\ldots, \epsilon^{\ind{m}} \le 0$ and $\epsilon^{\ind{m+1}},\ldots, \epsilon^{\ind{N}} > 0$.
\end{assumption}

In what follows, under Assumption \ref{a:nearly-constant-partitioned-weights},
we denote $c^{\ind{0}} = 0$ and $c^{\ind{i}} \defeq -\sum_{j=1}^i
\epsilon^{\ind{j}}$ for $i\in[N]$. Then we may write the distribution function
corresponding to $w^{\ind{1:N}}$ as follows:
$$
   F(i) = \sum_{j=1}^i w^{\ind{j}} = \frac{1}{N}\bigg(i - c^\ind{i}\bigg)\qquad\text{for}\qquad i=0,\ldots,N.
$$
\begin{lemma}
  \label{lem:lookup-characterisation}
  Under Assumption \ref{a:nearly-constant-partitioned-weights}:
  \begin{enumerate}[(i)]
    \item \label{item:monotonicity}
  $c^{\ind{1}} \le \cdots \le c^{\ind{m}}$, $c^{\ind{m}} > \cdots > c^{\ind{N}}$ and $c^{\ind{i}} \in [0,1)$ for $i\in[m]$.
  \item \label{item:lookupcharacterisation}
  For $u\in(0,1)$ and $\check{u}^{\ind{i}} \defeq (i-1+u)/N$, the following hold:
  \begin{align*}
    F(i-1) < \check{u}^i \le F(i) &\iff
    u \le 1 - c^{\ind{i}} \\
    F(i) < \check{u}^{i} \le F(i+1) &\iff
    u > 1 - c^{\ind{i}}.
  \end{align*}
\end{enumerate}
\end{lemma}
\begin{proof}
Because $\sum_i |\epsilon^{\ind{i}}| = 
\sum_i (\epsilon^{\ind{i}})_+ + \sum_i (-\epsilon^{\ind{i}})_+ < 2$,
and $\sum_i \epsilon^{\ind{i}} = 0$ so
$\sum_i (\epsilon^{\ind{i}})_+ = \sum_i (-\epsilon^{\ind{i}})_+ < 1$, from which \eqref{item:monotonicity} follows, and \eqref{item:lookupcharacterisation} is a direct consequence of $c^{i}\in[0,1)$.
\end{proof}

\begin{lemma}
  \label{lem:stratified-probability}
Let $A^{\ind{1:N}}$ be indices from stratified resampling (Definition \ref{def:stratified}). If Assumption \ref{a:nearly-constant-partitioned-weights} holds, then
$A^{\ind{i}}\in\{i,i+1\}$ for all $i\in[N]$ and for any $K\subset [N-1]$ and $S=[N]\setminus K$,
\[
  \P(A^{\ind{i}}=i,\, A^{\ind{j}}=j+1\, \text{for all $i\in S$ and $j\in K$}) 
  = \bigg(\prod_{j\in S} (1-c^{\ind{j}})\bigg) \bigg(\prod_{i\in K} c^{\ind{i}}\bigg).
\]
\end{lemma}
\begin{proof}
  Because the $\check{U_i}$'s are independent, we may write the probability of interest as
\[
  \bigg(\prod_{j\in S} \P\big(F(j-1) < \check{U}^{\ind{j}} \le F(j) \big) \bigg) \bigg(\prod_{i\in K} \P\big(F(i) < \check{U}^{\ind{i}} \le F(i+1) \big)\bigg),
\]
from which the result follows by Lemma \ref{lem:lookup-characterisation}.
\end{proof}

\begin{lemma}
  \label{lem:f-k-asymptotic-weights}
  Let $v^{\ind{1:N}}\ge 0$.
  The normalised weights $w_\Delta^{\ind{1:N}}$ corresponding to unnormalised weights $g_\Delta^{\ind{i}} = \exp(-\Delta v^{\ind{i}})$ may be written as
$$
   w_\Delta^{\ind{i}} 
   = \frac{1 + \epsilon^{\ind{i}}_\Delta}{N},\qquad\text{where}\qquad 
   \epsilon^{\ind{i}}_\Delta = \Delta (\bar{v} - v^{\ind{i}}) + r_\Delta^{\ind{i}}
$$
where $\bar{v} = N^{-1} \sum_{i=1}^N v^{\ind{i}}$ stands for the mean potential and the error terms $r_\Delta^{\ind{i}} = o(\Delta)$ and satisfy $|r_\Delta^{\ind{i}}| \le c \Delta$ for all $\Delta\in(0,1)$, where the constant $c$ depends only on $N$ and $\max_i v^{\ind{i}}$. Consequently,
$$
c_\Delta^{\ind{i}} \defeq - \sum_{j=1}^i \epsilon_\Delta^{\ind{j}} = \Delta \sum_{j=1}^i (v^{\ind{j}}- \bar{v} ) + \tilde{r}_\Delta^{\ind{i}},\qquad \tilde{r}_\Delta^{\ind{i}} = o(\Delta),\qquad 
|\tilde{r}_\Delta^{\ind{i}}| \le \tilde{c} \Delta.
$$
\end{lemma}
\begin{proof}
Direct calculation for $\epsilon_\Delta^{\ind{i}} = Ng_\Delta^{\ind{i}}/\sum_{j=1}^N g_\Delta^{\ind{j}} - 1$ yields that 
$$
\lim_{\Delta\to 0}\frac{d}{d\Delta}
\epsilon^{\ind{i}}_\Delta = \frac{1}{N} \sum_{j=1}^N v^{\ind{j}} - v^{\ind{i}} = \bar{v} - v^{\ind{i}},
$$
and properties of the error term 
can be verified, for instance, by using a Taylor expansion for the exponential function.
\end{proof}

\begin{proof}[Proof of Proposition \ref{prop:stratified-generator}]
Suppose first that $-v^{\ind{1:N}}$ are mean ordered, and that $\Delta$ is sufficiently small so that also $w_\Delta^{\ind{1:N}} \propto \exp(-\Delta v^{\ind{1:N}})$ satisfy Assumption \ref{a:nearly-constant-partitioned-weights}. Lemma \ref{lem:f-k-asymptotic-weights} together with Lemma \ref{lem:stratified-probability} give 
\[
  \lim_{\Delta \to 0+} \frac{1}{\Delta} r\big(a^{\ind{1:N}} \mid \exp(-\Delta v^{\ind{1:N}})\big) = \begin{cases}
    \sum_{j=1}^i (v^{\ind{j}} - \bar{v}), 
    & a^{\ind{i}} = i+1$ and $a^{\ind{\neg i}} = \neg i. \\
    0, & \text{otherwise.}
  \end{cases}
\]
The corresponding overall resampling rate is therefore
\[
  \sum_{i=1}^N \sum_{j=1}^i (v^{\ind{j}} - \bar{v})
  = \sum_{j=1}^N (N + 1 - j)(v^{\ind{j}} - \bar{v})
  = \sum_{j=1}^N j (\bar{v} - v^{\ind{j}}).
\]
The claim follows from this result applied to re-indexed $v_\varpi$ and $a_\varpi$.
\end{proof}

\begin{lemma}
  \label{lem:systematic-probability}
  Let $A^{\ind{1:N}}$ be indices from systematic resampling (Definition \ref{def:systematic}). If Assumption \ref{a:nearly-constant-partitioned-weights} holds, then for any $k\in[m]$ and $\ell \in [N] \setminus [m]$, 
  \begin{align*}
    &\P\big(A^{\ind{i}}=i\text{ for }i<k\text{ and } i\ge \ell,\; A^{\ind{j}}=j+1\text{ for } k\le j < \ell\big) \\
    &= \big(\min\{ c^{\ind{k}}, c^{\ind{\ell-1}}\} - \max\{ c^{\ind{k-1}}, c^{\ind{\ell}}\}\big)_+,
  \end{align*}
  and these events are the only possible in addition to the `no resampling' event, for which
  \begin{align*}
    &\P(A^{\ind{1:N}}=1{:}N) = 1 - c^{\ind{m}}. 
  \end{align*}
\end{lemma}
\begin{proof}
By Lemma \ref{lem:lookup-characterisation}, the event is equivalent to
\begin{align*}
  &\P\big( U\le 1 - c^{i}
  \text{ for }i<k\text{ and } i\ge \ell,\; U > 1 - c^{j}\text{ for } k\le j < \ell\big) \\
  &=\P\big(U \in (1 - c^{\ind{k}}, 1-c^{\ind{k-1}}], U \in (1 - c^{\ind{\ell-1}}, 1 -c^{\ind{\ell}}]\big),
\end{align*}
thanks to the monotonicity properties of $c^{\ind{1:m}}$ and $c^{\ind{m+1:N}}$. The latter follows similarly because $c^{\ind{m}} = \max_{i\in [N]} c^{\ind{i}}$.
\end{proof}

\begin{proof}[Proof of Proposition \ref{prop:systematic-generator}]
Suppose first that $-v^{\ind{1:N}}$ are mean ordered. Lemma \ref{lem:f-k-asymptotic-weights} with
Lemma \ref{lem:systematic-probability} yield 
\[
  \lim_{\Delta \to 0+} \frac{1}{\Delta} r\big(a^{\ind{1:N}} \mid \exp(-\Delta v^{\ind{1:N}})\big) = \begin{cases}
    \big(\min\{ s^{\ind{k}}, s^{\ind{\ell-1}} \} - \max\{ s^{\ind{k-1}}, s^{\ind{\ell}}\}\big)_+
    & a^{\ind{1:N}} = [k\to \ell]_N, \\
    0, & \text{otherwise},
  \end{cases}
\]
where $s^{\ind{0}}=0$ and $s^{\ind{i}} = \sum_{j=1}^i (v^{j} - \bar{v})$. The overall resampling rate is
$$
\sum_{k=1}^m \sum_{\ell=m+1}^N 
\big(\min\{ s^{\ind{k}}, s^{\ind{\ell-1}} \} - \max\{ s^{\ind{k-1}}, s^{\ind{\ell}}\}\big)_+ = \sum_{k=1}^m (s^{\ind{k}} - s^{\ind{k-1}}) = s^{\ind{m}},
$$
because $s^{\ind{0:m}}$ is increasing and $s^{\ind{m:N}}$ is decreasing. The result follows by re-indexing wrt.~$\varpi$.
\end{proof}

\begin{proposition}
  \label{prop:ssp-probabilities}
Suppose that the normalised weights $w^{\ind{1:N}}$ satisfy Assumption \ref{a:nearly-constant-partitioned-weights} and $\sum_{i=1}^N |\epsilon_i| < 1$. Then, for the SSP resampling with identity order $\varpi$, the only events with non-zero probability in addition to $1{:}N$ are of the form $A^{\ind{1:N}} = [k\to\ell]_N$, with probabilities:
$$
\P(A^{\ind{1:N}} = [k\to \ell]_N) =
\frac{(-\epsilon^{\ind{k}})_+ (\epsilon^{\ind{\ell}})_+}{\sum_{i=1}^N (\epsilon^{\ind{i}})_+},
\qquad \text{for }k,\ell\in[N].
$$
\end{proposition}
\begin{proof}
Thanks to Assumption \ref{a:nearly-constant-partitioned-weights}, the initial values of $p^{\ind{1:N}}$ satisfy $p^{\ind{i}} = 1 + \epsilon^{\ind{i}}$ for $i=1,\ldots,m$ and $p^{\ind{i}} = \epsilon^{\ind{j}}$ for $j=m+1,\ldots,N$, and now $c^{\ind{i}} = -\sum_{j=1}^i \epsilon^{\ind{j}}\in[0,1/2)$.

Note that the state of Algorithm \ref{alg:ssp} after lines \ref{line:set_deltas}--\ref{line:interchange_indices} is independent of the order of the indices $(i,j)$ before, so without loss of generality, we may assume that $i<j=k$ always before line \ref{line:set_deltas}.  We may deduce inductively that after line \ref{line:set_deltas} with $k \in \{2{:}m\}$:
\begin{itemize}
  \item $p^{\ind{i}} = 1 - c^{\ind{k-1}}$ and $p^{\ind{j}} = 1 + \epsilon^{\ind{k}}$, and
  \item $p^{\ind{i}}>1/2$ and $p^{\ind{j}} > 1/2$ and therefore $\delta^{\ind{i}} = 1 - p^{\ind{i}} = c^{\ind{k-1}}$ and $\delta^{\ind{j}} = 1-p^{\ind{j}} = -\epsilon^{\ind{k}}$.
\end{itemize}
With probability $\delta^{\ind{j}}/(\delta^{\ind{i}} + \delta^{\ind{j}}) = -\epsilon^{\ind{k}}/c^{\ind{k}}$, the indices at next iteration will be $i=k$ and $j=k+1$, and $r^{1:k-1}$ have all been incremented by one.
The probability to end up with indices $i=k$ and $j=m+1$ after iteration $m$ is $(-\epsilon^{\ind{k}}/c^{\ind{k}})\prod_{i=k+1}^m (c^{\ind{i-1}}/c^{\ind{i}}) = -\epsilon^{\ind{k}}/c^{\ind{m}}$, in which case all $r^{\ind{1:m}}$ have been incremented by one, except for $r^{\ind{k}}$.

Given the above scenario happens, then in the steps $k\in\{(m+1){:}(N-1)\}$ of the algorithm, it is again easy to see inductively that $p^{\ind{i}} + p^{\ind{j}} = 1 - c^{\ind{k}}< 1$ so $\delta^{\ind{i,j}} = p^{\ind{j,i}}$ and that $p^{\ind{i}} = 1 - c^{\ind{j-1}}$.
The probability to end up with $i=\ell$ and $j=N$ after iteration $N-1$ is therefore $\epsilon^{\ind{\ell}}/(1-c^{\ind{\ell}}) \prod_{j=\ell+1}^{N-1} (1-c^{\ind{j-1}})/(1-c^{\ind{j}}) = \epsilon^{\ind{\ell}}/(1-c^{\ind{N-1}})$, in which case in the beginning of the last step, $p^{\ind{i}}+p^{\ind{j}} = 1 - c^{\ind{N}} = 1$. Now, $r^{\ind{\ell}}$ will be incremented by one with probability $(1-c^{\ind{N-1}})$. We conclude the overall probability of outcome $[k\to \ell]_N$, which is equivalent to incrementing $r^{\ind{\ell}}$ and $r^{\ind{1:m}}$ by one except for $r^{\ind{k}}$.
\end{proof}

\begin{proof}[Proof of Theorem \ref{thm:resampling-order}]
  Note that for any $i$ such that $(\bar{v} - v^{\ind{i}})_+ > 0$, that is, $v^{\ind{i}} < \bar{v}$, we have
  $$ 
  (\bar{v} - v^{\ind{i}})_+ = \bar{v} - v^{\ind{i}} \leq \bar{v} - v_{\min}, 
  $$
  and there are of course at most $N - 1$ such $i$, so
  $$ \iota_{\mathrm{systematic}}^*(v^{\ind{1:N}}) \leq \#\{ i\,:\, v^{\ind{i}} < \bar{v}\} (\bar{v} - v_{\min} ) \leq (N-1) (\bar{v} - v_{\min} ) = \iota_{\mathrm{killing}}^*(v^{\ind{1:N}}).$$
  Assuming mean ordered $-v^{\ind{1:N}}$ we may write
  \begin{align*}
    \iota_{\mathrm{stratified}}^*(v^{\ind{1:N}}) - \iota_{\mathrm{systematic}}^*(v^{\ind{1:N}})
    &= \sum_{j=1}^{m} j (\bar{v} - v^{\ind{j}})  + \sum_{j=m+1}^{N} (j - 1) (\bar{v} - v^{\ind{j}}) \\
    &\ge m \sum_{j=1}^{m} (\bar{v} - v^{\ind{j}}) 
    + m \sum_{j=m+1}^{N} (\bar{v} - v^{\ind{j}}) = 0.
    \end{align*}
  
  To see that there cannot be such an order between $\iota^*_{\mathrm{killing}}$ and $\iota^*_{\mathrm{stratified}}$, consider $N = 3$ and strictly decreasing $v^{\ind{1:3}}$. Now,
  $v_{\min} = v^{\ind{3}}$ and 
  \begin{align*}
  \iota_{\mathrm{killing}}^*(v^{\ind{1:N}}) - \iota_{\mathrm{stratified}}^*(v^{\ind{1:N}}) & = 2(\bar{v} - v^{\ind{3}}) - (\bar{v} - v^{\ind{1}}) - 2(\bar{v}  - v^{\ind{2}}) - 3(\bar{v} - v^{\ind{3}}) = v^{\ind{2}} - \bar{v} \\
    & = \frac{2}{3} \Big( v^{\ind{2}} - \frac{v^{\ind{1}}+v^{\ind{3}}}2\Big),
  \end{align*}
  which can be positive or negative depending on $v^{\ind{2}} \in (v^{\ind{3}},v^{\ind{1}})$. A similar example can be constructed for any $N>3$ (we omit the details).
  \end{proof}

\section{Proof of Theorem \ref{thm:convergence}}
\label{app:convergence} 

This section is dedicated to the proof of Theorem \ref{thm:convergence}. Because of the technical nature of the theorem and its proof, we shall introduce some additional notation conventions specific to this section.

We use $(\real^d)^N$ and $\real^{dN}$ interchangeably for the state space of the particle filters $X^{\Delta_n}$, $n \in \N$ with the following identification: for $x \in \real^{dN}$, write $x^{\ind{j}}$, $j \in [N]$, for the vector $[x_{(j-1)d+1},x_{(j-1)d+2},\cdots,x_{jd-1},x_{jd}]^{\intercal} \in \real^d$, so that $x$ can be identified with $(x^{\ind{1}},\cdots,x^{\ind{N}}) \in (\real^d)^N$ and vice versa. In particular, \emph{superscripts $^{\ind{j}}$ refer to vector components and subscripts $_j$ to real coordinates}. This distinction will mostly be clear from context in the sequel, but we will emphasise it where necessary.

\subsection{A new construction of $X^\Delta_k$ in Theorem  \ref{thm:convergence} }
For notational convenience, we present a new (equivalent) construction for the $(\real^d)^N$-valued Markov chains $\{X^\Delta_k\}_{k\in \N_0}$ for $\Delta \in \{\Delta_n\,:\, n \in \N\}$. This new construction is self-contained in the sense that it does not make reference to Algorithm \ref{alg:pf} (which we did when introducing this chain in Section \ref{sec:convergence}).

Denote by $\mu \defeq M_1 \in \mathcal{P}(\real^d)$ the initial distribution for the particles, and by $W \defeq (W^\ind{1},\cdots,W^{\ind{N}})$ a fixed $dN$-dimensional Brownian motion (so that the $W^{\ind{j}}$'s are independent $d$-dimensional Brownian motions) with respect to a filtration $(\F_t)_{t \geq 0}$. We then redefine the Markov chain $X^\Delta \defeq (X^\Delta_k)_{k \in \N_0}$ by $(X^\Delta_0)^{\ind{j}} \sim \mu$ independently for $j \in [N]$ and
\[
  (X^\Delta_{k+1})^{\ind{j}} = (X^\Delta_{k})^{\ind{A_k(j)}} + b\big((X^\Delta_{k})^{\ind{A_k(j)}}\big) \Delta + \sigma\big((X^\Delta_{k})^{\ind{A_k(j)}}\big)\big( W^{\ind{j}}_{(k+1)\Delta} - W^{\ind{j}}_{k \Delta} \big)
\]
for $k \in \N_0$ and $j \in [N]$. Here $A_k \in [N]^N$ stands for the multi-index resulting from the resampling at time $k$, i.e.~$A_k \sim r(\cdot\mid\nu^\Delta(X^\Delta_k))$, and $A_k(j) \in [N]$ stands for the $j$'th index of $A_k$.

More precisely, we may write
\begin{equation}\label{eq:resampling-indices}
  A_k = \sum_{\ell = 1}^{N^N} \chrf\Big( \sum_{i=1}^{\ell-1} r(a_i\mid\nu^\Delta(X^\Delta_k)) < U^\Delta_k \leq \sum_{i=1}^{\ell} r(a_i\mid\nu^\Delta(X^\Delta_k) )\Big) a_\ell,
\end{equation}
where $\{a_1,\cdots,a_{N^N}\}$ is some fixed enumeration of $[N]^N$, and the $U^\Delta_k$'s are uniform random variables on $(0,1)$ independent of each other and the $W^{\ind{j}}$'s. We can take each $U^\Delta_k$ to be $\F_{k\Delta}$-measurable. The Markov chain $X^\Delta$ is then $(\F_{k\Delta})_{k\in\N_0}$-adapted for all $\Delta \in \{\Delta_n \,:\, n\in \N\}$.

Now define the continuous-time scaling $Z^\Delta$ of this Markov chain by
\begin{equation}\label{eq:z-scaling}
  Z^\Delta_t \defeq X^\Delta_{\lfloor t/\Delta \rfloor}, \quad t \in [0,\infty), 
\end{equation}
so that $Z^\Delta$ is for all $\Delta$ an $(\F_t)_{t\geq 0}$-adapted process with paths in $D_{\real^{dN}}[0,\infty)$, the Skorohod space of c\`adl\`ag paths in $\real^{dN}$. Recall the standard modulus of continuity of $D_{\real^{dN}}[0,\infty)$, defined for $z \in D_{\real^{dN}}[0,\infty)$, $\delta > 0$ and $t > 0$ by
\[
  \omega_{\delta,t}(z) \defeq \inf_{ \substack{0 = t_0 < t_1 < \cdots < t_{n-1} < t \leq t_n \\ t_i-t_{i-1} > \delta}} \, \max_{i\in[n]} \, \sup_{s,\,s' \in [t_{i-1},t_i)} |z(s') - z(s)|.
\]

\subsection{Outline of the proof of Theorem  \ref{thm:convergence} }
The main steps in the proof of parts (i) and (ii) of the theorem are as follows:
\begin{itemize}
\item In order to use a convergence result from \cite{ethier-kurtz}, our first goal is to show that the family of processes $(Z^{\Delta_n})_{n\in\N}$ is relatively compact with respect to convergence in distribution, i.e.~that  $\{ \P( Z^{\Delta_n} \in \cdot)\,:\, n \in \N \}$ is a relatively compact set in the weak topology of $\mathcal{P}(D_{\real^{dN}}[0,\infty))$. This is done in Proposition \ref{pr:relatively-compact} below.
\item The second step is to declare the continuous-time process $(Z_t)_{t\geq0}$ that the c\`adl\`ag extensions 
$( Z^{\Delta_n}_t )_{t\geq0}$  will converge to. The process $(Z_t)_{t\geq0}$ is the canonical c\`adl\`ag process $(\XX_t)_{t\geq0}$ introduced immediately before Proposition \ref{pr:well-posed}, equipped with a law possessing the generator $\L$. Proposition \ref{pr:well-posed} proves (by establishing the well-posedness of the corresponding martingale problem) that this law is uniquely determined and Markovian.
\item We then show in Proposition \ref{pr:generators-convergence} that the appropriately-scaled discrete-time derivative of the transition kernel of $(X^{\Delta_n}_k)_{k\in \N_0}$, the skeleton of $(Z^{\Delta_n}_t)_{t\geq0}$, converges in a suitable sense to the generator $\mathcal{L}$ as $n \to \infty$.
\item Proposition \ref{pr:convergence-in-distribution} and Theorem \ref{thm:general-convergence} then complete the proof of parts (i) and (ii).
\item The proof of part (iii) is presented in the end of the section.
\end{itemize}

\subsection{The proof}

\begin{proposition}
  \label{pr:relatively-compact} The family of processes $(Z^{\Delta_n})_{n \in \N}$ satisfies

(i) the following compact containment condition:
\[
  \lim_{C\to\infty} \sup_n \P\big(|Z_t^{\Delta_n}| > C\big) = 0
\]
for every $t > 0$;

(ii) the following uniform modulus of continuity condition:
\[
  \lim_{\delta\to0^+} \sup_n \P\big(\omega_{\delta,t}(Z^{\Delta_n}) > \varepsilon\big) = 0
\]
for every $t > 0$ and $\varepsilon > 0$.

Hence the family $(Z^{\Delta_n})_{n \in \N}$ is relatively compact with respect to convergence in distribution.
\end{proposition}

We refer to e.g.~\cite[Chapter 3]{ethier-kurtz} or \cite[Chapter 3]{billingsley} for details on the spaces $D_{\real^{dN}}[0,\infty)$ and $\mathcal{P}(D_{\real^{dN}}[0,\infty))$. In particular, see \cite[Chapter 3, Theorem 7.2]{ethier-kurtz} for why the conditions (i) and (ii) above together imply relative compactness. By \cite[Chapter 3, Remark 7.3]{ethier-kurtz}, relative compactness further implies the following stronger version of part (i):
\begin{equation}
  \label{eq:strong-compact-containment}
  \lim_{C\to\infty} \sup_n \P\big(\sup_{u \in [0,t]}|Z_u^{\Delta_n}| > C\big) = 0 \quad \text{for all} \quad t > 0.
\end{equation}

\begin{proof}
  (i) Fix $\Delta \in (\Delta_n)_{n\in\N}$. In order to avoid dealing with the jumps resulting from the resampling $r$ directly, we shall construct a larger tree of suitable discretisations of the underlying diffusion \eqref{eq:sde} which will at time $\lfloor t/\Delta\rfloor$ contain each of the $N$ particles of $X^\Delta_{\lfloor t/\Delta\rfloor}$ with sufficiently high probability -- this will be made precise in the argument below.

Denote by $\tau_k$, the time indices corresponding to the resampling events of the particles $X^\Delta$, i.e.
\[
  \tau_1 = \min \{ i \geq 0 \,:\, A_i \neq 1{:}N \}
\]
and inductively
\[
  \tau_{k+1} = \min \{ i > \tau_{k} \,:\, A_i \neq 1{:}N \}
\]
for $k \in \N$. These random times are obviously stopping times with respect to the filtration $(\F_{k\Delta})_{k \in \N_0}$ (and well-defined with probability $1$).

For multi-indices $a$ of length $|a| = 1$, i.e.~$a = (j)$ for $j \in [N]$, define the Markov chains $Y^{(j)}_k$ in $\real^d$ by
\[
  Y^{(j)}_0 = (X^\Delta_0)^{\ind{j}} \quad \textrm{and} \quad Y^{(j)}_{k+1} = Y^{(j)}_k + b(Y^{(j)}_k)\Delta + \sigma(Y^{(j)}_k) \big( W^{\ind{j}}_{(k+1)\Delta} - W^{\ind{j}}_{k \Delta} \big),
\]
so that $Y^{(j)}_k = (X^\Delta_k)^{\ind{j}}$ for $k \leq \tau_1$, and write
\[
  \mathcal{Y}^1 \defeq \bigl\{ (Y^a_k)_{k \geq 0} \,:\, |a| = 1 \bigr\}
\]
We then inductively define the trees $ \mathcal{Y}^{\ell+1} \defeq \{(Y^a_k)_{k \geq 0} \,:\, |a| = \ell+1 \}$ for integers $\ell > 0$ by
\[
  Y^a _k =
  \begin{cases}
    Y^{(a(1),\cdots,a(\ell))}_k \quad \text{for } k \leq \tau_\ell, \\
    Y^{a}_{k-1} + b(Y^{a}_{k-1})\Delta + \sigma(Y^{a}_{k-1}) \big( W^{\ind{a(\ell+1)}}_{k\Delta} - W^{\ind{a(\ell+1)}}_{(k-1) \Delta} \big) \quad \text{for } k > \tau_\ell;
  \end{cases}
\]
note in particular that for $Y^a \in \mathcal{Y}^{\ell+1}$, the states $Y^{(a(1),\cdots,a(\ell))}_k$ in the definition above come from $Y^{(a(1),\cdots,a(\ell))} \in \mathcal{Y}^{\ell}$. In other words, at time $\tau_\ell$ the chains in $\mathcal{Y}^\ell$ each branch into $N$ chains in $\mathcal{Y}^{\ell+1}$, evolving from time $\tau_\ell$ independently as Euler-Maruyama discretisations of the diffusion \eqref{eq:sde} driven by the $N$ Brownian motions driving the particles $(X^\Delta)^{1:N}$.

It is then easy to see from the construction that if there have been less than $\ell$ resampling events of the particles before time $\lfloor t /\Delta \rfloor$, i.e.
\[
  \sum_{k=0}^{\lfloor t /\Delta \rfloor - 1} \chrf\big( A_k \neq 1{:}N\big) < \ell,
\]
then $(X^\Delta_{\lfloor t/\Delta \rfloor})^{\ind{j}} \in \{ Y^a_{\lfloor t/\Delta \rfloor} \,:\, |a| = \ell\}$ for all $j \in [N]$, so that
\begin{equation}\label{eq:upotus}
  |Z_t^\Delta| \leq N \sum_{|a| = \ell}  |Y^a_{\lfloor t/\Delta \rfloor}| \leq c(\ell,N)\sqrt{\sum_{|a| = \ell} \big( |Y^a_{\lfloor t/\Delta \rfloor} - Y^a_0|^2 + |Y^a_0|^2 \big)},
\end{equation}
with some constant $c(\ell,N) < \infty$ independent of $t$ and $\Delta$. On the other hand, the probability of at least $\ell$ resampling events before time $\lfloor t /\Delta \rfloor$ can be controlled independently of $\Delta$ as follows:
\begin{align}
  \P\Big( \sum_{k=0}^{\lfloor t /\Delta \rfloor - 1} \chrf\big( A_k \neq 1{:}N\big) \geq \ell \Big) & \leq \frac1{\ell} \sum_{k=0}^{\lfloor t /\Delta \rfloor - 1} \P \big( A_k \neq 1{:}N \big) \notag \\
& \leq \frac{1}{\ell} \sum_{k=0}^{\lfloor t /\Delta \rfloor - 1} (c'\Delta) \leq \frac{c't}{\ell}, \label{eq:many-resamplings}
\end{align}
where the constant $c'$ is independent of $\Delta$, $t$ and $\ell$ (see \eqref{eq:resampling-generator} and \eqref{eq:resampling-indices}). Combining \eqref{eq:upotus} and \eqref{eq:many-resamplings} thus yields
\begin{align*}
  \P(|Z_t^\Delta| > C) & 
   \leq \P\Big( \sum_{|a| = \ell} \big( |Y^a_{\lfloor t/\Delta \rfloor} - Y^a_0|^2 + |Y^a_0|^2 \big) > \frac{C^2}{c(\ell,N)^2}\Big) + \frac{c't}{\ell}\\
  & \leq \P\Big( \sum_{|a| = \ell} |Y^a_{\lfloor t/\Delta \rfloor} - Y^a_0|^2 > \frac{C^2}{2c(\ell,N)^2}\Big) + \P\Big( \sum_{|a| = \ell} |Y^a_0|^2 > \frac{C^2}{2c(\ell,N)^2}\Big) + \frac{c't}{\ell} \\
  & \leq \frac{2c(\ell,N)^2}{C^2} \sum_{|a| = \ell} \E \big[|Y^a_{\lfloor t/\Delta \rfloor} - Y^a_0|^2\big] + \P\Big( \sum_{|a| = \ell} |Y^a_0|^2 > \frac{C^2}{2c(\ell,N)^2}\Big) + \frac{c't}{\ell}.
\end{align*}
Now the term $c't/\ell$ can be made arbitrarily small by a choice of $\ell$ independent of $\Delta$, and for fixed $\ell$ the term
\[
  \P\Big( \sum_{|a| = \ell} |Y^a_0|^2 > \frac{C^2}{2c(\ell,N)^2}\Big)
\]
obviously converges to zero as $C \to \infty$ uniformly in $\Delta$ (since each $Y^a_0$ is distributed as the initial distribution $\mu$). In other words, (i) follows if we can obtain an estimate for each $\E[|Y^a_{\lfloor t/\Delta \rfloor} - Y^a_0|^2]$ which may depend on $|a| = \ell$ and $t$ but not on $\Delta$.

To this end, note that the tree construction implies that
\begin{align}
  Y^a_{\lfloor t/\Delta \rfloor} - Y^a_0 & = \sum_{k=1}^{\lfloor t /\Delta \rfloor} \big( Y^a_{k} - Y^a_{k-1}\big) = \sum_{k=1}^{\lfloor t /\Delta \rfloor} \Big( b(Y^a_{k-1}) \Delta + \sigma(Y^a_{k-1}) \big( W^{\ind{j^a_{k-1}}}_{k\Delta} - W^{\ind{j^a_{k-1}}}_{(k-1)\Delta}\big)\Big) \notag \\
                                                   & = \sum_{k=1}^{\lfloor t /\Delta \rfloor} b(Y^a_{k-1})\Delta + \sum_{k=1}^{\lfloor t /\Delta \rfloor} \sigma(Y^a_{k-1}) \big( W^{\ind{j^a_{k-1}}}_{k \Delta} - W^{\ind{j^a_{k-1}}}_{(k-1)\Delta}\big), \label{eq:martingale1}
\end{align}
where each $j^a_k$ is an $\F_{k\Delta}$-measurable random index. Since $b$ is bounded, we simply get
\[
  \E\Big[ \big| \sum_{k=1}^{\lfloor t /\Delta \rfloor} b(Y^a_{k-1})\Delta \big|^2 \Big] \leq \big( \lfloor t/\Delta \rfloor \|b\|_{L^\infty} \Delta] \big)^2 \leq \|b\|_{L^\infty}^2 t^2.
\]
For the latter sum in \eqref{eq:martingale1}, note that
\begin{equation}\label{eq:martingale2}
  \sigma(Y^a_{k-1}) \big( W^{\ind{j^a_{k-1}}}_{k \Delta} - W^{\ind{j^a_{k-1}}}_{(k-1)\Delta}\big) = \sum_{j \in [N]} \chrf\big(j^a_{k-1} = j\big)\sigma(Y^a_{k-1}) \big( W^{\ind{j}}_{k \Delta} - W^{\ind{j}}_{(k-1)\Delta}\big),
\end{equation}
where the terms $\chrf(j^a_{k-1} = j)$ and $\sigma(Y^a_{k-1})$ are $\F_{(k-1)\Delta}$-measurable and the terms $W^{\ind{j}}_{k \Delta} - W^{\ind{j}}_{(k-1)\Delta}$ are distributed as $\mathcal{N}(0,\Delta I_{\real^d})$ independently of $\F_{(k-1)\Delta}$, so it is fairly easy to see that each real component of the $\real^d$-valued Markov chain
\[
  n \mapsto \sum_{k=1}^{n} \sigma(Y^a_{k-1}) \big( W^{\ind{j^a_{k-1}}}_{k \Delta} - W^{\ind{j^a_{k-1}}}_{(k-1)\Delta}\big)
\]
is a martingale, which leads to the estimate
\begin{align*}
  & \E \Big[ \big| \sum_{k=1}^{\lfloor t /\Delta \rfloor} \sigma(Y^a_{k-1}) \big( W^{\ind{j^a_{k-1}}}_{k \Delta} - W^{\ind{j^a_{k-1}}}_{(k-1)\Delta}\big) \big|^2 \Big] \\
& \qquad \qquad = \sum_{k=1}^{\lfloor t /\Delta \rfloor} \E \Big[ \big| \sigma(Y^a_{k-1}) \big( W^{\ind{j^a_{k-1}}}_{k \Delta} - W^{\ind{j^a_{k-1}}}_{(k-1)\Delta}\big) \big| ^2\Big] \\
  & \qquad \qquad  \leq \|\sigma\|_{L^\infty}^2 \sum_{k=1}^{\lfloor t /\Delta \rfloor} \E \Big[  \sum_{j \in [N]} \chrf\big(j^a_{k-1} = j\big)\big| W^{\ind{j}}_{k \Delta} - W^{\ind{j}}_{(k-1)\Delta}\big|^2 \Big] \\
  & \qquad \qquad = \|\sigma\|_{L^\infty}^2 \sum_{k=1}^{\lfloor t /\Delta \rfloor} \sum_{j \in [N]} \E \Big[ \chrf\big(j^a_{k-1} = j\big) \E \big[ \big| W^{\ind{j}}_{k \Delta} - W^{\ind{j}}_{(k-1)\Delta}\big|^2 \big| \F_{(k-1)\Delta}\big] \Big] \\
  & \qquad \qquad = \|\sigma\|_{L^\infty}^2 \lfloor t/\Delta \rfloor \Delta \leq \|\sigma\|_{L^\infty}^2 t.
\end{align*}
This finishes the proof of (i).

(ii) Fix $t > 0$, $\varepsilon > 0$ and $\gamma > 0$, with $\gamma$ arbitrarily small. We will show that
\begin{equation}\label{eq:limsup-gamma}
  \limsup_{\delta\to0^+} \Big( \sup_{\Delta \in (\Delta_n)_{n\in\N}} \P\big( \omega_{\delta,t}(Z^\Delta) > \varepsilon \big) \Big) \leq \gamma,
\end{equation}
which by the arbitrariness of $\gamma > 0$ implies (ii).

Consider parameters $0 < \Delta \leq \delta \ll 1$ (since obviously $\omega_{\delta,t}(Z^{\Delta'}) \equiv 0$ for $\Delta' > \delta$), where the condition ``$\ll 1$'' is quantified more precisely later on.
Like in \eqref{eq:many-resamplings} above, it will be convenient to disregard the possibility of arbitrarily many resampling events of the Markov chain $X^\Delta$ before the time $\lfloor t/\Delta\rfloor$:
\[
  \P \Big( \sum_{k=0}^{\lfloor t/\Delta\rfloor-1} \chrf\big(A_k \neq 1{:}N\big) \geq \ell \Big) \leq \frac{c't}{\ell},
\]
and $\ell \in \N$ can be taken so that the latter quantity is $< \gamma$. We can further limit our estimates to the case when the resampling events up to time $\tau_\ell$ happen happen sufficiently sparsely. Denoting $\tau_0 \equiv 0$, we have
\begin{align*}
  \P( \tau_k - \tau_{k-1} \leq m) & = \sum_{n} \P( \tau_k \leq m+n \mid \tau_{k-1} = n)\P(\tau_{k-1} = n) \\
  & = \sum_{n} \P\Big( \sum_{i=n+1}^{n+m} \chrf\big(A_i \neq 1{:}N\big) \geq 1 \mid \tau_{k-1} = n\Big)\P(\tau_{k-1} = n) \\
  & \leq \sum_{n} c'\Delta m \P(\tau_{k-1} = n) = c'\Delta m
\end{align*}
for any $k$, $m \in \N$, and so
\[
  \P\Big( \bigcup_{k=1}^{\ell} \{\tau_k - \tau_{k-1} \leq \lfloor 2\delta/\Delta\rfloor \}\Big) \leq 2\ell c'\delta.
\]
Thus the event
\begin{equation}\label{eq:sparse}
  B \defeq \Big\{ \sum_{k=0}^{\lfloor t/\Delta\rfloor-1}\chrf\big(A_k \neq 1{:}N\big) < \ell \Big\} \cap \bigcap_{k=1}^{\ell} \{\tau_k - \tau_{k-1} > \lfloor 2\delta/\Delta\rfloor \}
\end{equation}
has probability
\begin{equation}\label{eq:sparse-prob}
  \P(B) \geq 1 - \gamma - 2\ell c'\delta,
\end{equation}
which is sufficiently high for our purposes.

Now consider $\omega_{\delta,t}(Z^\Delta)$ within the event $B$. Denote by $\tau^*$ the smallest resampling time $\tau_k$ such that $\tau_k \geq \lfloor t/\Delta \rfloor$. By \eqref{eq:sparse} we have $\tau^* \leq \tau_\ell$ and $(\tau_k - \tau_{k-1})\Delta > \lfloor 2\delta/\Delta\rfloor\Delta > \delta$ for $\tau_k \leq \tau^*$. Obviously $(\tau^* + 1)\Delta \geq t$. Thus the partition
\begin{equation}\label{eq:partition-unoptimal}
  0 < (\tau_1 + 1)\Delta < (\tau_2 + 1)\Delta < \cdots < (\tau^* + 1)\Delta
\end{equation}
is a valid candidate for the infimum in the definition of $\omega_{\delta,t}$.
However since some of the interval lengths $(\tau_k - \tau_{k-1})\Delta$ may be
unnecessarily long for estimating $\omega_{\delta,t}(Z^\Delta)$, we shall
refine the partition \eqref{eq:partition-unoptimal} as follows. Let
\[
  0 = t_0 < t_1 < \cdots < t_{m-1} < t \leq t_m
\]
be such that $\{ (\tau_k + 1)\Delta \,:\, \tau_k < \tau^* \} \subset \{t_i\,:\, i < m\}$, $t_i - t_{i-1} \in (\delta,2\delta]$ and $t_m \leq \min( (\tau^*+1)\Delta, t + 2\delta)$. To see why this is possible, simply divide any interval of the form $[(\tau_{k-1}+1)\Delta,(\tau_k+1)\Delta)$ with length $>2\delta$ into sufficiently many smaller subintervals.

Note that we do not claim that the path-wise choice of the partition $\{t_i\}$ in $B$ is in any way measurable, but this will ultimately not be an issue below. The partition will be used as a stepping stone for a path-wise upper bound for $\omega_{\delta,t}(Z^\Delta)$ (for paths in $B$) which \emph{is} measurable.

The point of the above construction is that jumps induced by the resampling will not happen on continuous-time intervals of the type $[t_{i-1},t_i)$, as is easily seen from the original partition \eqref{eq:partition-unoptimal}. Thus
\begin{align*}
  \omega_{\delta,t}(Z^\Delta)
  & \leq \max_{i \in [m] } \, \sup_{s,\,s' \in [t_{i-1},t_i) } |Z^\Delta_{s'} - Z^\Delta_s|
   = \max_{i \in [m] } \, \max_{\lfloor \frac{t_{i-1}}{\Delta}\rfloor \leq k,\,k' < \lfloor \frac{t_{i}}{\Delta}\rfloor} |X^\Delta_{k'} - X^\Delta_k| \\
  & \leq \sum_{j \in [N] }\max_{i \in [m] } \, \max_{\lfloor \frac{t_{i-1}}{\Delta}\rfloor \leq k,\,k' < \lfloor \frac{t_{i}}{\Delta}\rfloor} |(X^\Delta_{k'})^{\ind{j}} - (X^\Delta_k)^{\ind{j}}| \\
  & = \sum_{j \in [N] } \max_{i \in [m] } \, \max_{\lfloor \frac{t_{i-1}}{\Delta}\rfloor \leq k < k' < \lfloor \frac{t_{i}}{\Delta}\rfloor} \big|\sum_{k \leq i <k'} b( (X^\Delta_i)^{\ind{j}})\Delta + \sum_{k \leq i < k'} \sigma((X^\Delta_i)^{\ind{j}})\big(W^{\ind{j}}_{(i+1)\Delta} - W^{\ind{j}}_{i\Delta}\big) \big|.
\end{align*}
The first inner sum can be estimated as follows:
\begin{align*}
  \big|\sum_{k \leq i <k'} b( (X^\Delta_i)^{\ind{j}})\Delta\big|
  & \leq (k' - k) \|b\|_{L^\infty} \Delta  \leq \|b\|_{L^\infty} \Big( \Big\lfloor\frac{t_{i}}{\Delta}\Big\rfloor - \Big\lfloor\frac{t_{i-1}}{\Delta}\Big\rfloor \Big) \Delta  \\
  & \leq \|b\|_{L^\infty} \Big( \frac{t_{i}}{\Delta} - \frac{t_{i-1}}{\Delta} + 1\Big) \Delta \leq 3\|b\|_{L^\infty} \delta,
\end{align*}
and here we specify the assumption (already implicitly made above) that $\delta < \varepsilon/(6N\|b\|_{L^\infty})$. Applying this estimate to the one above yields
\[
  \omega_{\delta,t}(Z^\Delta) \leq \frac\varepsilon{2} + \sum_{j \in [N] } \max_{i \in [m] } \, \max_{\lfloor \frac{t_{i-1}}{\Delta}\rfloor \leq k < k' < \lfloor \frac{t_{i}}{\Delta}\rfloor} \big|\sum_{k \leq i < k'} \sigma((X^\Delta_i)^{\ind{j}})\big(W^{\ind{j}}_{(i+1)\Delta} - W^{\ind{j}}_{i\Delta}\big) \big|.
\]
To sidestep the measurability issues arising from the path-wise choice of the partition $\{t_i\}$, note that for any $i \in [m]$ we can find $n \in \{0, \cdots, \lceil t/\delta\rceil \}$ such that
\[
  \Big\{ \Big\lfloor \frac{t_{i-1}}{{\Delta}}\Big\rfloor,  \Big\lfloor \frac{t_{i-1}}{\Delta}\Big\rfloor +1 , \cdots ,\Big\lfloor \frac{t_{i}}{\Delta}\Big\rfloor - 1 \Big\}
  \subset
  \Big\{n \Big\lceil \frac{\delta}{\Delta}\Big\rceil, n \Big\lceil \frac{\delta}{\Delta}\Big\rceil + 1, \cdots , (n+3) \Big\lceil \frac{\delta}{\Delta}\Big\rceil -1\Big\},
\]
so
\begin{align*}
  \omega_{\delta,t}(Z^\Delta) & \leq \frac\varepsilon{2} + \sum_{j \in [N] } \max_{0 \leq n \leq \lceil \frac{t}{\delta} \rceil } \, \, \max_{n\lceil\frac{\delta}{\Delta}\rceil \leq k<k'<(n+3)\lceil\frac{\delta}{\Delta}\rceil} \big|\sum_{k \leq i < k'} \sigma((X^\Delta_i)^{\ind{j}})\big(W^{\ind{j}}_{(i+1)\Delta} - W^{\ind{j}}_{i\Delta}\big) \big| \\
  & \leq \frac\varepsilon{2} + 2 \sum_{j \in [N] } \max_{0 \leq n \leq \lceil\frac{t}{\delta} \rceil } \, \, \underbrace{\max_{n\lceil\frac{\delta}{\Delta}\rceil < k <(n+3)\lceil\frac{\delta}{\Delta}\rceil} \big|\sum_{n\lceil\frac{\delta}{\Delta}\rceil \leq i < k} \sigma((X^\Delta_i)^{\ind{j}})\big(W^{\ind{j}}_{(i+1)\Delta} - W^{\ind{j}}_{i\Delta}\big) \big|}_{=: M^{\ind{j}}_n}.
\end{align*}
Combining this with \eqref{eq:sparse} and \eqref{eq:sparse-prob} thus yields
\begin{align}
  \P\big( \omega_{\delta,t}(Z^\Delta) > \varepsilon \big) & \leq \P\Big( \frac\varepsilon{2} + 2 \sum_{j \in [N] } \max_{0 \leq n \leq \lceil\frac{t}{\delta} \rceil }M^{\ind{j}}_n > \varepsilon \Big) + 2\ell c'\delta + \gamma \notag\\
  & \leq \sum_{j\in[N]} \P\Big( \max_{0 \leq n \leq \lceil\frac{t}{\delta} \rceil }M^{\ind{j}}_n > \frac{\varepsilon}{4N} \Big) + 2\ell c'\delta + \gamma \notag \\
  & \leq \sum_{j\in[N]} \sum_{0 \leq n \leq \lceil t/\delta\rceil} \P\Big(M^{\ind{j}}_n > \frac{\varepsilon}{4N} \Big) + 2\ell c'\delta + \gamma. \label{eq:maximal}
\end{align}

Then in order to estimate a term of the form $\P(M^{\ind{j}}_n > \frac\varepsilon{4N})$, recall from the discussion following \eqref{eq:martingale2} that each real component of the $\real^d$-valued Markov chain
\[
 k \mapsto S^{\ind{j}}_{n,k} \defeq \sum_{n\lceil\frac{\delta}{\Delta}\rceil \leq i < k} \sigma((X^\Delta_i)^{\ind{j}})\big(W^{\ind{j}}_{(i+1)\Delta} - W^{\ind{j}}_{i\Delta}\big), \qquad k \geq n\Big\lceil\frac{\delta}{\Delta}\Big\rceil,
\]
is a martingale starting from zero (i.e.~the sum above is interpreted as zero for $k = n\lceil\delta/\Delta\rceil$). Write $[S^{\ind{j}}_n]_k$ for
\[
  \sum_{n\lceil\frac{\delta}{\Delta}\rceil \leq i < k} \big| \sigma((X^\Delta_i)^{\ind{j}})\big(W^{\ind{j}}_{(i+1)\Delta} - W^{\ind{j}}_{i\Delta}\big) \big|^2 ,
\]
which is the sum of the quadratic variations of the real components of $S^{\ind{j}}_{n,\cdot}$ up to time $k > n\lceil\delta/\Delta\rceil$. Then, for arbitrary fixed $p > 2$, we may use the Burkholder-Davis-Gundy inequality (see e.g.~\cite[pp.~499]{shiryaev} or \cite[Theorem 18.7]{kallenberg}) and H\"older's inequality to obtain
\begin{align*}
  \P\Big(M^{\ind{j}}_n > \frac{\varepsilon}{4N} \Big) 
  & \lesssim \varepsilon^{-p} \E \big[ |M^{\ind{j}}_n|^p \big] \lesssim \varepsilon^{-p} \E \big[ [S^{\ind{j}}_n]_{(n+3)\lceil\frac\delta\Delta\rceil}^{p/2}\big] \\
  & \lesssim \varepsilon^{-p} \E\Big[ \Big( \sum_{n\lceil\frac{\delta}{\Delta}\rceil \leq i < (n+3)\lceil\frac{\delta}{\Delta}\rceil} \big| \sigma((X^\Delta_i)^{\ind{j}})\big(W^{\ind{j}}_{(i+1)\Delta} - W^{\ind{j}}_{i\Delta}\big) \big|^2 \Big)^{p/2} \Big] \\
  & \lesssim \varepsilon^{-p} \|\sigma\|_{L^\infty}^p \Big\lceil\frac\delta\Delta\Big\rceil^{(p-2)/2} \E \Big[ \sum_{n\lceil\frac{\delta}{\Delta}\rceil \leq i < (n+3)\lceil\frac{\delta}{\Delta}\rceil} \big| W^{\ind{j}}_{(i+1)\Delta} - W^{\ind{j}}_{i\Delta}\big|^p \Big] \\
  & \lesssim \varepsilon^{-p} \|\sigma\|_{L^\infty}^p \Big\lceil\frac\delta\Delta\Big\rceil^{p/2} \Delta^{p/2} \lesssim \varepsilon^{-p} \|\sigma\|_{L^\infty}^p \delta^{p/2},
\end{align*}
where the implicit multiplicative constant in each inequality is independent of $\delta$ and $\Delta$ (but can of course depend on $p$, $N$ and the dimension $d$ of the state space of the particles). Applying this estimate to \eqref{eq:maximal} gives
\[
  \P\big( \omega_{\delta,t}(Z^\Delta) > \varepsilon \big) \leq c'' t \varepsilon^{-p} \|\sigma\|_{L^\infty}^p \delta^{(p/2)-1} + 2\ell c'\delta + \gamma
\]
with constants $c'$ and $c''$ independent of $\delta$ and $\Delta$, and since $(p/2)-1>0$, this yields \eqref{eq:limsup-gamma} and thus finishes the proof of (ii).
\end{proof}

Our next step is to verify that the infinitesimal generator $\L$ in Theorem \ref{thm:convergence} is associated with a \emph{well-posed martingale problem}. Let us briefly recall the concept of martingale problems. 

Denote by $\XX$ the canonical c\`adl\`ag process, given by the projections
\[
  D_{\real^{dN}}[0,\infty) \owns z \mapsto \XX_t(z) \defeq z(t) \in \real^{dN},
\]
and by $\FF$ the filtration generated by $\XX$. The \emph{martingale problem} for $(\L,\testf(\real^{dN}))$ concerns the existence of a probability measure $\P_\eta \in \mathcal{P}(D_{\real^{dN}}[0,\infty))$ for any given $\eta \in \mathcal{P}(\real^{dN})$ such that
\[
  (t,z) \mapsto f(\XX_t(z)) - f(\XX_0(z)) - \int_0^t \L f(\XX_u(z)) \ud u
\]
is a martingale (with respect to $\FF$) under $\P_\eta$ for any $f \in \testf(\real^{dN})$, and $\P_\eta \circ (\XX_0)^{-1} = \eta$. The martingale problem is said to be \emph{well-posed} if for every $\eta$ a solution $\P_\eta$ exists and is unique. To be more precise, ``uniqueness'' here means uniqueness in terms of finite-dimensional distributions -- we refer to \cite[Chapter 4]{ethier-kurtz} for a thorough examination of this subject.

In order to write the generator $\L$ explicitly, we slightly abuse notation and define the functions $b^*\colon \real^{dN} \to \real^{dN}$ and $\sigma^*\colon \real^{dN} \to \real^{dN} \times \real^{dN}$ for $x \defeq (x^{\ind{1}},x^{\ind{2}},\cdots,x^{\ind{N}}) \in (\real^{d})^N$ by
\[
   b_i^*(x) = \big( b(x^{\ind{j}}) \big)_{i - (j-1)d} \qquad \text{for} \qquad j \in [N] \text{ and } (j-1)d < i \leq j d
\]
and
\[
  \sigma^*_{i,j}(x) =
  \begin{cases}
    \big( \sigma(x^{\ind{k}})\big)_{i-(k-1)d , j-(k-1)d} \quad & \text{if} \quad (k-1)d < i,\, j \leq k d \text{ for some } k\in[N], \\
    0 & \text{otherwise}.
  \end{cases}
\]
Recall that the subscripts $_i$ are to be interpreted as real coordinates $\real^d$ or $\real^{dN}$, the subscripts $_{i,j}$ as real entries of a matrix in $\real^d\times \real^d$ or $\real^{dN} \times \real^{dN}$ and the superscripts $^{\ind{j}}$ as $\real^d$-components of $(\real^d)^N$, keeping in line with our previous notation. The functions $b^*$ and $\sigma^*$ obviously have the same properties $b$ and $\sigma$, i.e.~they are Lipschitz continuous and bounded, and $\sigma^*$ is uniformly non-degenerate in the sense of \eqref{eq:nondegenerate}.

We then have
\begin{align}
  \L f(x) 
  & = \Big( \sum_{i=1}^{dN} b^*_i(x) \partial_i f(x) + \frac12\sum_{i,j=1}^{dN} \big(\sigma^*(x)^\intercal \sigma^*(x) \big)_{i,j} \partial_{i,j}f(x) \Big) \notag \\
  & \quad \qquad \qquad + \sum_{a\in[N]^N \setminus\{1{:}N\}} \iota^a (x) \big( f(x^{\ind{a(1{:}N)}}) - f(x)\big) \notag \\
    & =: \Lmut f(x) + \Ljump f(x) \label{eq:mut-jump}
\end{align}
for $f \in C^{\infty}_c(\real^{dN})$, where the jump intensity functions $\iota^a \colon\real^{dN}\to[0,\infty)$ were assumed to be bounded and continuous. Recall that $x^{\ind{a(1{:}N)}}$ above stands for $(x^{\ind{a(1)}},\cdots,x^{\ind{a(N)}})$.

\begin{proposition}
  \label{pr:well-posed}
  The martingale problem for $(\L,\testf(\real^{dN}))$ is well-posed.
\end{proposition}

\begin{proof}
  Write $\testf \defeq \testf(\real^{dN})$. A fairly general existence result by W.~Hoh, see \cite[Theorem 3.15]{hoh} or \cite[Theorem 3.24]{bottcher-schilling-wang} and the references therein, implies that the martingale problem for $(\L,\testf)$ with our assumptions for the coefficient functions has a solution for any initial distribution $\eta$.

In order to verify uniqueness, we first note that our assumptions on the coefficient functions $b^*$ and $\sigma^*$ imply that $\Lmut$ is a standard non-degenerate (in the sense of \eqref{eq:nondegenerate}) diffusion-type generator, so the martingale problem for $(\Lmut,\testf)$ is well-posed, and in fact $\Lmut$ generates a Feller process -- see e.g.~\cite[Chapter 8, Section 1]{ethier-kurtz} and \cite[Theorem 17.24]{kallenberg}.

Then consider cutoff functions $\chi^k \in \testf$, $k \in \N$, such that $0 \leq \chi^k \leq 1$ and $\chi^k(x) = 1$ for $|x| \leq k$. A standard perturbation result implies that each 
$\L^k \defeq \Lmut + \chi^k \Ljump$, i.e.
\[
  \L^kf(x) = \Lmut f(x) + \chi^k(x) \sum_{a\in[N]^N\setminus\{1{:}N \}} \iota^a (x) \big( f(x^{\ind{a(1{:}N)}}) - f(x)\big),
\]
is a Feller generator and that well-posedness holds for the martingale problem for $(\L^k,\testf)$; see e.g.~\cite[Chapter 1, Theorem 7.1]{ethier-kurtz} or \cite[Theorem 2.8.1]{jacob}.
The cutoff function $\chi^k$ is needed here to ensure that $\L^k$ maps functions in $\testf$ to continuous functions that vanish at infinity, which is not necessarily the case for $\L$ itself since the jump intensity functions $\iota^a$ are merely continuous and bounded.

Finally, since the martingale problem for $(\L, \testf)$ has a solution for any initial distribution as noted above, and $\L$ always coincides locally with some $\L^k$, a standard localisation procedure for the well-posedness of martingale problems \cite[Chapter 4, Theorems 6.1 and 6.2]{ethier-kurtz} yields well-posedness for the martingale problem for $(\L, \testf)$.
\end{proof}

Finally, let us check that appropriately-scaled discrete generators of the Markov chains $X^{\Delta_n}$ converge to $\L$.

\begin{proposition}
  \label{pr:generators-convergence}
For $n \in \N$, write $\Gamma_n$ for the transition kernel of $X^{\Delta_n}$. Then
\[
  \L_n f(x) \defeq \frac{1}{\Delta_n} \Big( \int_{\real^{dN}} f(y) \Gamma_{n}(x,\ud y) - f(x) \Big) \stackrel{n\to\infty}{\longrightarrow} \L f(x)
\]
for all $f \in \testf(\real^{dN})$ with bounded and pointwise convergence with respect to $x \in \real^{dN}$, i.e.~$\L_n f$ is uniformly bounded for all $f$.
\end{proposition}

\begin{proof}
  Let $\Delta \defeq \Delta_n$ for some $n$. With the notation introduced above, we have
\begin{align*}
  X^\Delta_{k+1} & = \sum_{a \in [N]^N} \chrf\big( A_k = a \big) \Big( (X^\Delta_k)^{\ind{a(1{:}N)}} + b^*\big((X^\Delta_k)^{\ind{a(1{:}N)}}\big)\Delta \\
  & \qquad  \qquad + \sigma^*\big((X^\Delta_k)^{\ind{a(1{:}N)}}\big)(W_{(k+1)\Delta} - W_{k\Delta}) \Big),
\end{align*}
where $W_t \defeq (W^1_t,\cdots,W^N_t)$. Thus, $\Gamma_n$ can be expressed as
\[
  \int_{\real^{dN}}f(y)\Gamma_n(x,\ud y) = \E\Big[ \sum_{a \in [N]^N } \chrf\big( A^\Delta(x) = a \big) f \big( x^{\ind{a(1{:}N)}} + b^*(x^{\ind{a(1{:}N)}})\Delta + \sigma^*(x^{\ind{a(1{:}N)}})B^\Delta \big) \Big]
\]
for $x \in \real^{dN}$ and bounded and Borel measurable $f\colon\real^{dN}\to\real$, where
\[
  A^\Delta(x) = \sum_{\ell = 1}^{N^N} \chrf\Big( \sum_{i=1}^{\ell-1} r(a_i\mid\nu^\Delta(x)) < U \leq \sum_{i=1}^{\ell} r(a_i\mid\nu^\Delta(x) )\Big) a_\ell
\]
for some $U \sim \mathcal{U}(0,1)$ and $B^\Delta \sim \mathcal{N}(0,\Delta I_{\real^{dN}})$ independent from $U$.

A standard application of Taylor's theorem then justifies the following calculations for $f \in \testf(\real^{dN})$:
\begin{align}
  & \int_{\real^{dN}} f(y) \Gamma_{n}(x,\ud y) - f(x) \notag \\
  & \quad = \int_{\real^{dN}} \big(f(y) -f(x)\big) \Gamma_{n}(x,\ud y) \notag \\
  & \quad = \sum_{a\in[N]^N} \E \Big[ \chrf  \big(A^\Delta(x) = a\big)\big( f\big(  x^{\ind{a(1{:}N)}} + b^*(x^{\ind{a(1{:}N)}})\Delta \notag \\
  & \qquad \qquad \qquad + \sigma^*(x^{\ind{a(1{:}N)}})B^\Delta  \big) - f(x) \big) \Big]\notag \\
  & \quad = r\big( 1{:}N \mid \nu^\Delta(x)\big)\big( \Lmut f(x)\Delta + o(\Delta) \big) \notag\\
  & \qquad \qquad \qquad + \sum_{a \in [N]^N\setminus\{ 1{:}N \}} r(a|\nu^\Delta(x))\big( f(x^{\ind{a(1{:}N)}}) - f(x) + O(\Delta)\big). \label{eq:discrete-generator1}
\end{align}
To be more precise, note that if $a = a_\ell$, then within the event $\{A^\Delta(x) = a\}$ the term
\[
  f\big(  x^{\ind{a(1{:}N)}} + b^*(x^{\ind{a(1{:}N)}})\Delta + \sigma^*(x^{\ind{a(1{:}N)}})B^\Delta  \big)
\]
can be estimated by a sufficiently regular sum of terms of the form
$
  g(x^{\ind{a(1{:}N)}})h(B^\Delta),
$
and
\begin{align*}
  & \E\Big[ \chrf  \big(A^\Delta(x) = a\big) g(x^{\ind{a(1{:}N)}})h(B^\Delta) \Big] \\
  & \quad = g( x^{\ind{a(1{:}N)}}) \E\Big[ \chrf \Big( \sum_{i=1}^{\ell-1} r(a_i\mid\nu^\Delta(x)) < U \leq \sum_{i=1}^{\ell} r(a_i\mid\nu^\Delta(x) )\Big) h(B^\Delta) \Big] \\
  & \quad = g( x^{\ind{a(1{:}N)}}) \E\Big[ \chrf  \Big( \sum_{i=1}^{\ell-1} r(a_i\mid\nu^\Delta(x)) < U \leq \sum_{i=1}^{\ell} r(a_i\mid\nu^\Delta(x) )\Big) \Big] \E [ h(B^\Delta) ] \\
  & \quad = g( x^{\ind{a(1{:}N)}}) r(a \mid\nu^\Delta(x)) \E [ h(B^\Delta) ]
\end{align*}
for each such term, and so the calculations and estimates for each multi-index can be carried out as one would when computing the generator a standard diffusion process.

In particular, from \eqref{eq:discrete-generator1} and the assumption \eqref{eq:resampling-generator} we get
\[
  \lim_{n\to\infty} \frac1{\Delta_n} \Big(\int_{\real^{dN}} f(y) \Gamma_n(x,\ud y) - f(x)\Big) = \Lmut f(x) + \Ljump f(x) = \L f(x),
\]
with bounded and pointwise convergence with respect to $x$.
\end{proof}

From the last three auxiliary results we obtain

\begin{proposition}
  \label{pr:convergence-in-distribution}
Let $Z$ be the solution to the martingale problem for $(\L,\testf(\real^{dN}))$ with initial distribution
\[
  \mu \times \mu \times \cdots \times \mu \in \mathcal{P}(\real^{dN})
\]
(see Proposition \ref{pr:well-posed}). Then
\[
  \lim_{n\to\infty} Z^{\Delta_n} = Z
\]
in distribution, i.e.~
\[
  \lim_{n\to\infty} \E[ F(Z^{\Delta_n} )] = \E[ F(Z)]
\]
for all bounded and continuous $F\colon D_{\real^{dN}}[0,\infty)\to\real$.
\end{proposition}

\begin{proof}
  We aim to apply \cite[Corollary 8.13, Chapter 4]{ethier-kurtz}, which requires us to verify a number of conditions. The notation for the auxiliary processes below, $\xi^f_n$ and $\varphi^f_n$, follows closely the statement of said result.
  
  Recall first that the family $(Z^{\Delta_n})_{n \in \N}$ is relatively compact in the sense of Proposition \ref{pr:relatively-compact}, and that the function space $\testf(\real^{dN})$ is \emph{separating} in the sense that for $P$, $Q \in \mathcal{P}(\real^{dN})$,
\[
  \int f\, \ud P = \int f\, \ud Q \quad \forall f\in \testf(\real^{dN})
\]
implies $P = Q$. For $f \in \testf(\real^{dN})$ and $t \geq 0$, write $\xi^f_n(t) \defeq f(Z^{\Delta_n}_t)$ and $\varphi^f_n(t) \defeq \L_n f(Z^{\Delta_n}_t)$. Then we obviously have
\[
  \sup_n \sup_{s \leq t} \E [ |\xi^f_n(s) |] < \infty \qquad \text{and} \qquad \sup_n \sup_{s \leq t} \E [ |\varphi^f_n(s) |] < \infty
\]
for all $f \in \testf(\real^{dN})$ and $t > 0$,
\[
  \lim_{n\to\infty} \E \Big[ \big(\xi^f_n(t) - f(Z^{\Delta_n}_t) \big) \prod_{i=1}^k h_i\big( Z^{\Delta_n}_{t_i}\big)\Big] = 0
\]
for all $f \in \testf(\real^{dN})$, $\{h_1,\cdots,h_k\} \subset \testf(\real^{dN})$ and $0 \leq t_1 < \cdots < t_k \leq t$ (trivially by the definition of $\xi^f_n$), and
\[
  \lim_{n\to\infty}\E \Big[ \big(\varphi^f_n(t) - \L f(Z^{\Delta_n}_t) \big) \prod_{i=1}^k h_i\big( Z^{\Delta_n}_{t_i}\big)\Big] = 0
\]
for $f$, $h_i$ and $t_i$ like above by dominated convergence via the bounded and pointwise convergence of the $\L_n f's$ (see Proposition \ref{pr:generators-convergence}). Thus \cite[Corollary 8.13, Chapter 4]{ethier-kurtz} implies the desired convergence. \qedhere
\end{proof}

Theorem \ref{thm:convergence} still does not automatically follow from this, since functions of finite-dimensional distributions are in general \emph{not} continuous on $D_{\real^{dN}}[0,\infty)$. Let us thus formulate the following convergence result which essentially encapsulates the three parts of Theorem \ref{thm:convergence}.

\begin{theorem}
  \label{thm:general-convergence}
Let $Z$ be the c\`adl\`ag process in Proposition \ref{pr:convergence-in-distribution}. Then
\begin{equation}\label{eq:finite-dimensional2}
  \lim_{n\to\infty}\E[ f(Z^{\Delta_n}_{t_1},\cdots,Z^{\Delta_n}_{t_T}) F(Z^{\Delta_n}) ] = \E[ f(Z_{t_1},\cdots,Z_{t_T}) F(Z) ]
\end{equation}
for all finite $\{t_1,\cdots,t_T\} \subset [0,\infty)$ and bounded and continuous functions $f\colon (\real^{dN})^T \to \real$ and $F \colon D_{\real^{dN}}[0,\infty) \to \real$.
\end{theorem}

\begin{proof}
We first establish \eqref{eq:finite-dimensional2} with the additional assumption that $f$ is Lipschitz continuous. To this end, note that first that
\[
  z \mapsto \frac1\delta \int_0^\delta f(z_{t_1+h},\cdots,z_{t_T+h}) dh \cdot F(z) =: Q^\delta(z) 
\]
is for any $\delta > 0$ a continuous function on the Skorohod space $D_{\real^{dN}}[0,\infty)$, and by the c\`adl\`ag property
\[
  \lim_{\delta\to0^+} Q^\delta(z) = f(z_{t_1},\cdots,z_{t_T}) F(z) =: Q(z)
\]
with bounded and pointwise convergence with respect to $z \in D_{\real^{dN}}[0,\infty)$. Thus,
\begin{align}
  |\E[Q(Z)] - \E[Q(Z^{\Delta_n})]| \leq & |\E[Q(Z) - Q^\delta(Z)]| + |\E[Q^\delta(Z) - Q^\delta(Z^{\Delta_n})]| \notag \\
  & \qquad + |\E[Q^\delta(Z^{\Delta_n}) - Q(Z^{\Delta_n})]|. \label{eq:finite-dimensional3}
\end{align}
By the dominated convergence theorem, the first term on the right-hand side of \eqref{eq:finite-dimensional3} can be taken arbitrarily small by considering any sufficiently small $\delta > 0$. For a fixed $\delta > 0$, the second term is arbitrarily small for small enough $\Delta_n \in (0,\delta)$ by the convergence in distribution of the processes. For the third term, note that
\begin{align*}
  |\E[Q^\delta(Z^{\Delta_n}) - Q(Z^{\Delta_n})]| & \leq c_{F} \E\Big[ \frac1\delta \int_0^\delta |f(Z^{\Delta_n}_{t_1+h},\cdots,Z^{\Delta_n}_{t_T+h}) -f(Z^{\Delta_n}_{t_1},\cdots,Z^{\Delta_n}_{t_T}) |dh \Big] \\
  & \leq c_{f,F} \sum_{i=1}^T \E\big[ \sup_{t_i < s < t_{i}+\delta} |Z^{\Delta_n}_s - Z^{\Delta_n}_t | \land 1 \big] \\
  & \leq c_{f,F} \sum_{i=1}^T \Big( \E\big[ \sup_{t_i < s < t_{i}+\delta} |Z^{\Delta_n}_s - Z^{\Delta_n}_t |^2 \land 1 \big] \Big)^{1/2}.
\end{align*}
The expectation in the $i$th term above can be estimated by
\[
  \E\Big[ \max_{ \lfloor \frac{t_i}{\Delta_n} \rfloor < k \leq \lfloor \frac{t_i+\delta}{\Delta_n} \rfloor} | X^{\Delta_n}_k -  X^{\Delta_n}_{\lfloor t_i/\Delta_n \rfloor}|^2\land 1 \Big],
\]
and we may estimate the latter quantity uniformly in $\Delta_n < \delta$ in a familiar manner: the probability of any resampling-induced jumps of $X^{\Delta_n}$ between the indices in the inner maximum is at most of order $\delta$, and outside of this event, the term
\[
  \max_{ \lfloor \frac{t_i}{\Delta_n} \rfloor < k \leq \lfloor \frac{t_i+\delta}{\Delta_n} \rfloor} | X^{\Delta_n}_k -  X^{\Delta_n}_{\lfloor t_i/\Delta_n \rfloor}|^2
\]
can be estimated using a martingale decomposition and the Burkholder-Davis-Gundy inequality (with $p=2$ for the sake of simplicity) like in the proof part (ii) of Proposition \ref{pr:relatively-compact}, resulting in an upper bound of order $\delta$.

This finishes the proof of \eqref{eq:finite-dimensional2} for bounded and Lipschitz continuous $f$. The convergence then extends to bounded and uniformly continuous $f$ by a standard $\varepsilon/3$-argument, since any such function can be approximated uniformly by Lipschitz continuous functions.

Finally, if $f\colon (\real^{dN})^T\to\real$ is a bounded and continuous function, the compact containment condition (part (i) of Proposition \ref{pr:relatively-compact}) implies that there is a compact set $K \subset (\real^{dN})^T$ such that
$
  (Z^{\Delta_n}_{t_1},\cdots,Z^{\Delta_n}_{t_T}) \in K
$
and
$
  (Z_{t_1},\cdots,Z_{t_T}) \in K
$
with probability arbitrarily close to $1$. Subsequently there is a uniformly continuous $g\colon (\real^{dN})^T\to\real$ such that $g \equiv f$ on $K$ and $\|g\|_{L^\infty} \leq \|f\|_{L^\infty}$. Thus an $\varepsilon/3$-argument again establishes \eqref{eq:finite-dimensional2} for $f$.
\end{proof}

\begin{proof}[Proof of Theorem \ref{thm:convergence}]
  Parts (i) and (ii) follow automatically by taking $F \equiv 1$ in Theorem \ref{thm:general-convergence} above.

For part (iii), note first that for fixed $\tau > 0$ the function
\[
  D_{\real^{dN}}[0,\infty) \owns z \mapsto F(z) \defeq \exp\big(-\int_0^\tau \mathcal{V}(u,z_u) \ud u \big)
\]
is continuous, so 
\begin{equation}\label{eq:fk-convergence}
  \lim_{n\to\infty} \E[f(Z^{\Delta_n}_\tau) F(Z^{\Delta_n})] = \E[f(Z_\tau) F(Z)]
\end{equation}
by Theorem \ref{thm:general-convergence}. On the other hand, using the simple estimate $|e^{-x} - e^{-y}| \leq |x-y|\land 1$ for $x$, $y \geq 0$, we can estimate
\begin{align}
  & \Big| \E[f(Z^{\Delta_n}_\tau) F(Z^{\Delta_n})] - \E \Big[ f(X^{\Delta_n}_{\lfloor \tau/\Delta_n\rfloor} ) \prod_{k=0}^{\lfloor \tau/\Delta_n\rfloor-1} e^{- \Delta_n \mathcal{V}( k\Delta_n, X^{\Delta_n}_k )}\Big] \Big| \notag\\
  & \qquad = \E\Big[ |f(Z^{\Delta_n}_\tau)| \big| F(Z^{\Delta_n}) - \exp\big( -\Delta_n \sum_{k=0}^{\lfloor \tau/\Delta_n\rfloor-1} \mathcal{V}( k\Delta_n, Z^{\Delta_n}_{k\Delta_n} \big) \big| \Big] \notag \\
  & \qquad \leq \|f\|_{L^{\infty}} \E\Big[ \big| \int_{0}^{\tau} \mathcal{V}(u,Z^{\Delta_n}_u) \ud u - \Delta_n \sum_{k=0}^{\lfloor \tau/\Delta_n\rfloor-1} \mathcal{V}(k\Delta_n,Z^{\Delta_n}_{k\Delta_n})\big| \land 1 \Big]. \label{eq:fk-convergence2}
\end{align}
Now since $\mathcal{V}$ is uniformly continuous on compact subsets of $[0,\infty) \times \real^{dN}$ and the paths $(Z^{\Delta_n}_u)_{0 \leq u \leq \tau}$ stay inside some compact subset of $\real^{dN}$ with arbitrarily high probability (see the discussion after the statement of Proposition \ref{pr:relatively-compact}, in particular \eqref{eq:strong-compact-containment}), it is fairly easy to see that \eqref{eq:fk-convergence2} converges to zero as $n \to \infty$. Combining this with \eqref{eq:fk-convergence} above yields the desired convergence.
\end{proof}

\section{Proof of Theorem \ref{thm:fk-unbiased}}
\label{app:fk-unbiased} 

We start with the following auxiliary result, which states the intuitively simple fact that although the sample paths of the particle filter $Z$ are discontinuous with probability $1$, the probability of discontinuities (i.e.~resampling-induced jumps) at any given time is negligible.

\begin{proposition}\label{pr:left-continuous}
Let $Z$ be the c\`adl\`ag process in Theorem \ref{thm:convergence}. Then
\[
  \P\big( Z_t = Z_{t-} \big) = 1
\]
for all $t > 0$.
\end{proposition}

\begin{proof}
We only give a brief outline of a proof. It suffices to show that
\begin{equation}\label{eq:left-continuous}
  \E\big[ | f(Z_t) - f(Z_{t-})| \big] = 0
\end{equation}
for any bounded and Lipschitz continuous $f\colon \real^{dN}\to \real$. The c\`adl\`ag property implies that
\[
  f(z_t) - f(z_{t-}) =  \lim_{\delta\to 0^+} \frac1\delta\Big( \int_{t}^{t+\delta} f(z_u) \ud u - \int_{t-\delta}^{t} f(z_u) \ud u \Big)
\]
for all $z \in D_{\real^{dN}}[0,\infty)$ and $t > 0$ with bounded and pointwise convergence, and the expression inside the limit is for each $\delta > 0$ a continuous function of $z \in D_{\real^{dN}}[0,\infty)$. We can then proceed as in the proof of Theorem \ref{thm:general-convergence}, noting that resampling-induced jumps of the processes $Z^{\Delta_n}$ happen with arbitrarily small probability on arbitrarily small time intervals.
\end{proof}

\begin{proof}[Proof of Theorem \ref{thm:fk-unbiased}]
It suffices to establish \eqref{eq:fk-unbiased} for $f \in \testf(\real^d)$,
since any bounded and measurable function on $\real^d$ can be approximated pointwise by an uniformly bounded sequence of functions in $\testf(\real^d)$.

  Write 
\[
  \QQ_t(F) \defeq \E\Big[ F(Z_t) \exp\Big(-\int_0^t \overline{V}(Z_u) \ud u\Big)\Big]
\]
for $F \in \testf(\real^{dN})$ and
\[
  Q_t(f) \defeq \E\Big[ f(z_t) \exp\Big(-\int_0^t V(z_u) \ud u\Big)\Big]
\]
for $f \in \testf(\real^d)$. 

For fixed $F$, the measure flow $\QQ_t(F)$ is differentiable with respect to $t > 0$. In order to show this and compute $\frac{\ud}{\ud t} \QQ_t(F)$, write $E_t \defeq \exp(-\int_0^t \overline{V}(Z_u) \ud u)$ for $t > 0$. Then for $\delta > 0$ with $\delta \ll 1$,
\[
  F(Z_{t+\delta}) E_{t+\delta} - F(Z_t) E_{t} = E_t\big( F(Z_{t+\delta}) - F(Z_t) \big) + F(Z_{t+\delta}) \big( E_{t+\delta} - E_t\big).
\]
By the martingale problem (see Appendix \ref{app:convergence}), $F(Z_{t+\delta}) - F(Z_t)$ can be written as $M_{t+\delta} - M_t + \int_t^{t+\delta} \L F(Z_u) \ud u$ for some martingale $M$ (with respect to the filtration $\F^Z$ generated by $Z$). Using this in combination with the tower property of conditional expectations (with respect to $\F^Z_t$), we can calculate
\[
  \QQ_{t+\delta}(F) - \QQ_t (F) = \E\Big[ E_t \int_t^{t+\delta} \L F(Z_u) \ud u \Big] + \E\big[ F(Z_{t+\delta}) \big( E_{t+\delta} - E_t\big)\big].
\]
The c\`adl\`ag property in conjunction with the dominated convergence theorem then implies
\[
  \lim_{\delta\to0^+}\frac{\QQ_{t+\delta}(F) - \QQ_t (F)}{\delta} = \E[ \L F(Z_t) E_t ] - \E[ F(Z_t) \overline{V}(Z_t) E_t] = \QQ_t\big( (\L - \overline{V})F \big).
\]
For negative $\delta$, we may compute $\QQ_{t+\delta}(F) - \QQ_t (F)$ in a similar manner by using the tower property with respect to the filtration $\F^Z_{t+\delta}$ instead of $\F^Z_t$. Then Proposition \ref{pr:left-continuous} above together with the dominated convergence theorem imply
\begin{equation}\label{eq:particle-evolution}
  \lim_{\delta\to 0^-}\frac{\QQ_{t+\delta}(F) - \QQ_t (F)}{\delta} = \QQ_t\big( (\L - \overline{V})F \big).
\end{equation}
In a similar (in fact easier since the sample paths of $z$ are automatically continuous) way we may calculate
\begin{equation}\label{eq:sde-evolution}
  \frac{\ud }{\ud t} Q_t(f) = Q_t\big( (L - V)f \big)
\end{equation}
for any $f \in \testf(\real^d)$, where $L$ is the infinitesimal generator corresponding to the diffusion \eqref{eq:sde}.

Now the left-hand side of the statement of the Theorem is $\QQ_t(\overline{f}) =: \hat{Q}_t(f)$ and the right-hand side is $Q_t(f)$. We will show that the evolution equation for $\hat{Q}$ is of the same form as \eqref{eq:sde-evolution}. To this end, recall that $\L = \Lmut + \Ljump$ as in \eqref{eq:mut-jump}. For $f \in \testf(\real^d)$, we simply get $\Lmut(\overline{f}) = \overline{Lf}$, and further
\begin{align*}
  \Ljump(\overline{f})(x) & = \sum_{a \neq 1{:}N} \iota^a(x) \Big( \frac1N \sum_{i = 1}^N f(x^{\ind{a(i)}}) - \frac1N \sum_{i = 1}^N f(x^{\ind{i}})\Big) \\
  & = \frac1N \sum_{i = 1}^N \Big(\sum_{a \neq 1{:}N} \iota^a(x) (\#_{a,i} - 1)\Big)f(x^{\ind{i}}),
\end{align*}
where $\#_{a,i} \defeq \#\{j \,:\, a(j) = i\}$. Thus, comparing the right-hand sides of \eqref{eq:particle-evolution} and \eqref{eq:sde-evolution} (with $\overline{f}$ in place of $F$), we see that
\[
  \frac{\ud }{\ud t} \hat{Q}_t(f) =  \hat{Q}_t\big( (L - V)f \big)
\]
is equivalent to
\begin{equation}\label{eq:generator-unbiased2}
  \QQ_t\Big( \frac1N \sum_{i = 1}^N \Big(\sum_{a \neq 1{:}N} \iota^a(x) (\#_{a,i} - 1)\Big)f(x^{\ind{i}}) \Big) = \QQ_t( \overline{V}\cdot\overline{f} - \overline{Vf} ).
\end{equation}
Writing out the expression inside the right-hand side parentheses and comparing the coefficients of the $f(x^{\ind{i}})$'s, one sees that \eqref{eq:generator-unbiased2} will follow from the identity
\begin{equation}\label{eq:generator-unbiased}
  \sum_{a \neq 1{:}N} \iota^a(x) (\#_{a,i} - 1) = \overline{V}(x) - V(x^{\ind{i}})
\end{equation}
for all $i$ and $x$, which is simply assumption (i) of the Theorem.

Thus the measure flows $Q$ and $\hat{Q}$ satisfy the same evolution equation, and by the assumptions we have $Q_0 = \hat{Q}_0$. Our next task is to verify that this evolution equation is well-posed in a suitable sense.

In order to work in a space of probability measures, let us denote by $\dot{\real}^d$ the standard one-point compactification of $\real^d$ with infinity point $o$. Define the probability measures $Q_t^o \in \mathcal{P}(\dot{\real}^d)$ and $\hat{Q}_t^o \in \mathcal{P}(\dot{\real}^d)$, $t \geq 0$, by
\[
  Q_t^o(f) \defeq Q_t(f_{|\real^d}) + \big(1 - Q_t(1_{|\real^d})\big) f(o) = Q_t\big((f-f(o))_{|\real^d}\big) + f(o)
\]
for (bounded and) continuous $f \colon \dot{\real}^d \to \real$ and similarly for $\hat{Q}$ in place of $Q$. Define $M$ as the collection of continuous functions $f$ on $\dot{\real}^d$ such that
\[
  \big(f - f(o)\big)_{|\real^d} \in \testf(\real^d),
\]
and define the linear operator $A$ on $M$ by
\[
  Af(x) = L\big(f - f(o)\big)_{|\real^d}(x) + V(x) \big( f(o) - f(x) \big), \quad x \in \dot{\real}^d,
\]
with the understanding that $Af(o) = 0$.

Now the flows $Q^o$ and $\hat{Q}^o$ both satisfy the forward equation
\begin{equation}\label{eq:forward-equation}
  \mu_t(f) = \mu_0(f) + \int_0^t \mu_u(Af) \ud u, \quad t > 0,
\end{equation}
for (the natural extensions of) $f \in \testf(\real^d)$, and it is easy to see that this extends to $f \in M$. We are thus in a place to apply a uniqueness result from \cite{ethier-kurtz}: it is routinely verified that $A$ satisfies the positive maximum principle, $M$ is an algebra of functions that is dense in the space of continuous functions on $\dot{\real}^d$ (with respect to the $\sup$-norm), and the $D_{\dot{\real}^d}[0,\infty)$-martingale problem for $(A,M)$ is well-posed (see Appendix \ref{app:convergence}). Thus \cite[Chapter 4, Proposition 9.19]{ethier-kurtz} yields well-posedness for the forward equation \eqref{eq:forward-equation}.

In particular, $Q^o_t(f) = \hat{Q}^o_t(f)$ for all $t > 0$ and (extensions of) $f \in \testf(\real^d)$, which translates to $Q_t(f) = \hat{Q}_t(f)$.
\end{proof}

Finally, let us prove Proposition \ref{pr:asymptotic-unbiased}:

\begin{proof}[Proof of Proposition \ref{pr:asymptotic-unbiased}]
  Let $i \in [N]$, $v^{\ind{1{:}N}} \in [0,\infty)^N$ and $\Delta > 0$. By Assumption \ref{a:unbiased-resampling},
\begin{align*}
  \frac{N e^{-\Delta v^{\ind{i}}}}   {\sum_{j = 1}^N e^{-\Delta v^{\ind{j}}}} & = \sum_{j = 1}^N r\big( a(j) = i \mid \exp(-\Delta v^{\ind{1{:}N}})\big)
  = \sum_{j = 1}^N \sum_{a \in [N]^N } r\big( a \mid \exp(-\Delta v^{\ind{1{:}N}}) \big) \charfun{a(j) = i } \\
  & = \sum_{a\in[N]^N} r\big( a \mid \exp(-\Delta v^{\ind{1{:}N}}) \big) \#\{j \,:\, a(j) = i\},
\end{align*}
and subtracting $\sum_{a\in[N]^N} r(a\mid \exp(-\Delta v^{\ind{1{:}N}}) ) \equiv 1$ from this yields
\[
  \sum_{a \neq 1{:}N} r\big( a \mid \exp(-\Delta v^{\ind{1{:}N}})\big) \big(\#\{j \,:\, a(j) = i\} - 1\big) = \frac{N e^{-\Delta v^{\ind{i}}} - \sum_{j = 1}^N e^{-\Delta v^{\ind{j}}} }{\sum_{j = 1}^N e^{-\Delta v^{\ind{j}}} }.
\]
Dividing this by $\Delta$ and taking $\Delta\to 0^+$ leads to \eqref{eq:asymptotic-unbiased}, as in the proof of Lemma \ref{lem:f-k-asymptotic-weights}.
\end{proof}

\end{appendix}

\section*{Supplementary material}

Source codes for the experiments are available at \url{https://github.com/mvihola/weakly-informative-resampling-codes}

\section*{Acknowledgments}

TS and MV were supported by Academy of Finland grant 315619 and the Finnish
Centre of Excellence in Randomness and Structures.  The authors wish to
acknowledge CSC --- IT Center for Science, Finland, for computational resources.

\end{document}